\documentclass[11pt,a4paper]{article}
\usepackage{amsmath,amsfonts}
\usepackage{algorithmic}
\usepackage{algorithm}
\usepackage{array}
\usepackage{float} 
\usepackage[caption=false,font=normalsize,labelfont=sf,textfont=sf]{subfig}
\usepackage{textcomp}
\usepackage{stfloats}
\usepackage{url}
\usepackage{verbatim}
\usepackage{graphicx}
\usepackage{cite}

 \usepackage{caption}
\usepackage[T1]{fontenc}
\usepackage[utf8]{inputenc}
\usepackage{authblk}

\usepackage{amssymb,amsfonts}

\usepackage{graphicx}
\usepackage{booktabs}
\usepackage{pifont}
\usepackage{bm}
\usepackage{dutchcal}
\usepackage{gensymb}
\usepackage{geometry}
\usepackage{makecell}
\usepackage{amssymb}
\usepackage{array}
\usepackage{booktabs}
\usepackage{multirow}
\usepackage{subfig}
\usepackage{graphicx}
\usepackage{amsmath}
\usepackage{epstopdf}
\usepackage{ntheorem}
\usepackage{tcolorbox}
\tcbuselibrary{theorems}

\usepackage[normalem]{ulem}

\theorembodyfont{\upshape}
\newtheorem{thm}{Theorem}
\newtheorem{lem}{Lemma}

\newtheorem*{proof}{Proof}

\DeclareMathOperator{\EMST}{EMST} 

\begin{document}
\newtcolorbox{Metcalfe}{width=2cm,size=small,
	colframe=white,
	colback=yellow!45!white}
\newtcolorbox{Sarnoff}{width=2cm,size=small,
	colframe=white,
	colback=blue!10!white}
\newtcolorbox{Odlyzko}{width=2cm,size=small,
	colframe=white,
	colback=green!10!white}

\title{Emergence of Metcalfe's Law: Mechanism and Model}

%

\author[a,b]{Cheng Wang \thanks{Corresponding author: cwang@tongji.edu.cn}}
\author[b]{Yi Wang}
\author[a,b]{Changjun Jiang}

\affil[a]{Shanghai Artificial Intelligence Laboratory, China}
\affil[b]{Department of Computer Science and Technology, Tongji University, China}

\renewcommand*{\Affilfont}{\small\it} 
\date{} 


\maketitle

\begin{abstract}
Metcalfe's Law captures the relationship between the value of a network and its scale, asserting that a network's value is directly proportional to the square of its size. Over the past four decades, various researchers have proposed different scaling laws on this subject. Remarkably, these seemingly conflicting conclusions have all been substantiated by robust data validation, raising the question of which law holds greater representativeness. Consequently, there remains a need for inherent mechanism to underpin these laws. This study aims to bridge this disparity by offering a theoretical interpretation of Metcalfe's Law and its variations.
Based on a certain degree of consensus that "traffic is value", network effects are gauged by means of network traffic load. A general analytical boundary for network traffic load is deduced by balancing practicality and analytical feasibility through the establishment of a comprehensive network model. From this foundation, the mechanism behind Metcalfe's Law and its variants is elucidated, aligning the theoretical derivations with the previously validated empirical evidence for Metcalfe's Law.
\end{abstract}

{\noindent}	 \rule[2pt]{17cm}{0.05em}\\

Among the various influential ideas that emerged during the Internet boom, Metcalfe's Law \cite{gilder1993metcalf} stood out as one of the most significant notions. 
This law proposes that the value of a network is directly proportional to the square of the number of its nodes, i.e., $V\propto n^2$, following a clear scaling law. 
It is widely regarded as a universal empirical rule and has achieved prominent status.
Similar to Moore's Law, it is considered as one of the five fundamental empirical laws that have stood the test of time \cite{ross2003commandments}. 
In 1996, Reed Hundt, the former chairman of the Federal Communications Commission, asserted that Metcalfe's Law and Moore's Law provided the foundation for comprehension of the Internet \cite{kaprowski1996fcc}. 
With the emergence of Metcalfe's Law, more research has surfaced. 
Studies of Metcalfe's Law in the context of network effects have become a contentious subject in various explorations \cite{ reed1999weapon,swann2002functional,briscoe2006metcalfe, madureira2013empirical, metcalfe2013metcalfe, van2014metcalfe,zhang2015tencent}. 
With the emergence of Metcalfe's Law, a series of other laws have also been proposed \cite{reed1999weapon, swann2002functional, briscoe2006metcalfe,metcalfe2013metcalfe, zhang2015tencent, wang2023revisiting}. These laws can be concluded as Reed's Law \cite{reed1999weapon}, Sarnoff's Law \cite{swann2002functional}, and Odlyzko's Law \cite{briscoe2006metcalfe}. These formulation can be showed as $V\propto 2^n$, $V\propto n$ and $V\propto n\log n$. As all of these laws describe the relationship between network value and scale, we refer to them as variants of Metcalfe's Law.

Metcalfe's Law is a key explanatory model for technology growth, from traditional communication to modern web applications and social networks\cite{hendler2008metcalfe}. It significantly influences digital blockchain networks and cryptocurrency valuations, including Bitcoin analysis\cite{alabi2017digital,peterson2018metcalfe,wheatley2019bitcoin}. Jun et al.\cite{jun2021how} showed social network services follow a polynomial adoption curve based on web traffic. Scala et al.\cite{scala2023explosive} highlighted networks' explosive growth value. Ma et al.\cite{ma2023does} examined Metcalfe's Law's nonlinear impact on the green economy's digitization. Tang et al.\cite{tang5study} introduced the DEVA method for enterprise valuation using Metcalfe's Law. Assif et al.\cite{assif2023fair} analyzed crowdsourcing network value distribution fairness. Roberto\cite{moro2023smart} applied the law to Covid-19 healthcare digital transformation. Nowak et al.\cite{nowak2023social} integrated network laws into urban platform coordination.

A series of data validation studies \cite{hendler2008metcalfe,peterson2018metcalfe, wheatley2019bitcoin,alabi2017digital,reed1999weapon, swann2002functional,  zhang2015tencent } have confirmed seemingly conflicting versions of Metcalfe's Law and its variants, creating a debate over which scaling law holds greater representativeness. Additionally, in 2023, Zhang et al. \cite{zhang2023facebook} reevaluated the most recent data from Facebook and Tencent, revealing that a cube law ($V\propto n^3$) offered a more precise representation of actual network values compared to Metcalfe's Law.

The reason why these laws always hold true in the data validation work, and what laws can better describe the relationship between the value of the modeling network and the number of nodes, are unresolved questions in the current research. Such questions cannot be solved by empirical research. In order to properly understand the relationship between network value and size, and to make people realize the true universality of this research, a series of mechanisms need to be established to reproduce these laws.

In this research, we utilize the measurement of network traffic as a fundamental metric for assessing the value of networks within the scope of the mechanism modeling \cite{rajgopal2001web,jun2021how}, aligning with the principle that \textit{traffic is value}.
Expanding on this insight, we rely on network traffic load as a crucial indicator to evaluate the worth of networks. The reasoning behind this choice lies in the fact that a higher traffic load signifies an increased network capacity to accommodate nodes, consequently enhancing the overall value of the network.  

\section*{Problem Formulation}\label{Pro_For}
\vspace{-0.1in}

Here we report on the generation of network and the mechanisms based on the network. We consider a general network model, where all nodes can be supposed as both sources and users of data, capable of generating and receiving data.
We outline this network structure and define the traffic load accordingly.
We propose a network architecture consisting of two layers, i.e., the underlying and overlying layer, as shown in Figure 1(a). The former is dedicated to the physical transmission of data, while the latter facilitates interactions among devices or users. 
This research emphasizes the examination of information transmission load specifically within the overlying layer, exploring the intricacies and effects of this load on network performance and overall system behavior.
By analyzing the information transmission load between overlying layer, we can gain a deeper insight into the transmission patterns of network architecture.


The true value of a network lies in fulfilling the requirements of its constituent nodes. Unlike an isolated node that operates independently, the purpose of network is to facilitate the satisfaction of node needs through data transmission. As data is transmitted, it generates traffic, which becomes a crucial indicator of the value of network. Thus, the value of network becomes evident through the volume of network traffic.
As Rajgopal et al. \cite{rajgopal2001web} found that the relationship between traffic and network economy is related to Metcalfe's law, we use \textit{traffic load} as a metric of network value.

Given a network $\mathrm{N}$, the traffic load for a data transportation session  $\mathrm{T}_{k}$ (e.g., broadcast) over $\mathrm{N}$ are defined as
$\mathrm{L}_{\mathrm{N}}^{\mathrm{T}_k}=\lambda_k\cdot d_{{\mathrm{S}_\mathrm{N}},u_k},$ where $\lambda_{k}$ is the data arrival rate of node $u_{k}$, and $d_{{\mathrm{S}_\mathrm{N}},u_k}$ indicates the aggregated distance data transporting from $u_{k}$ to destinations of $\mathrm{T}_{k}$ over the network $\mathrm{N}$ under a given data transportation scheme $\mathrm{S}_\mathrm{N}$. The detailed description is showed in Methods \ref{Def-Tra}.

\begin{figure}[htbp]
	\begin{minipage}{0.47\columnwidth}
		\centering
		\subfloat[The proposed two-layer network model]{
			\includegraphics[width=1.1\linewidth]{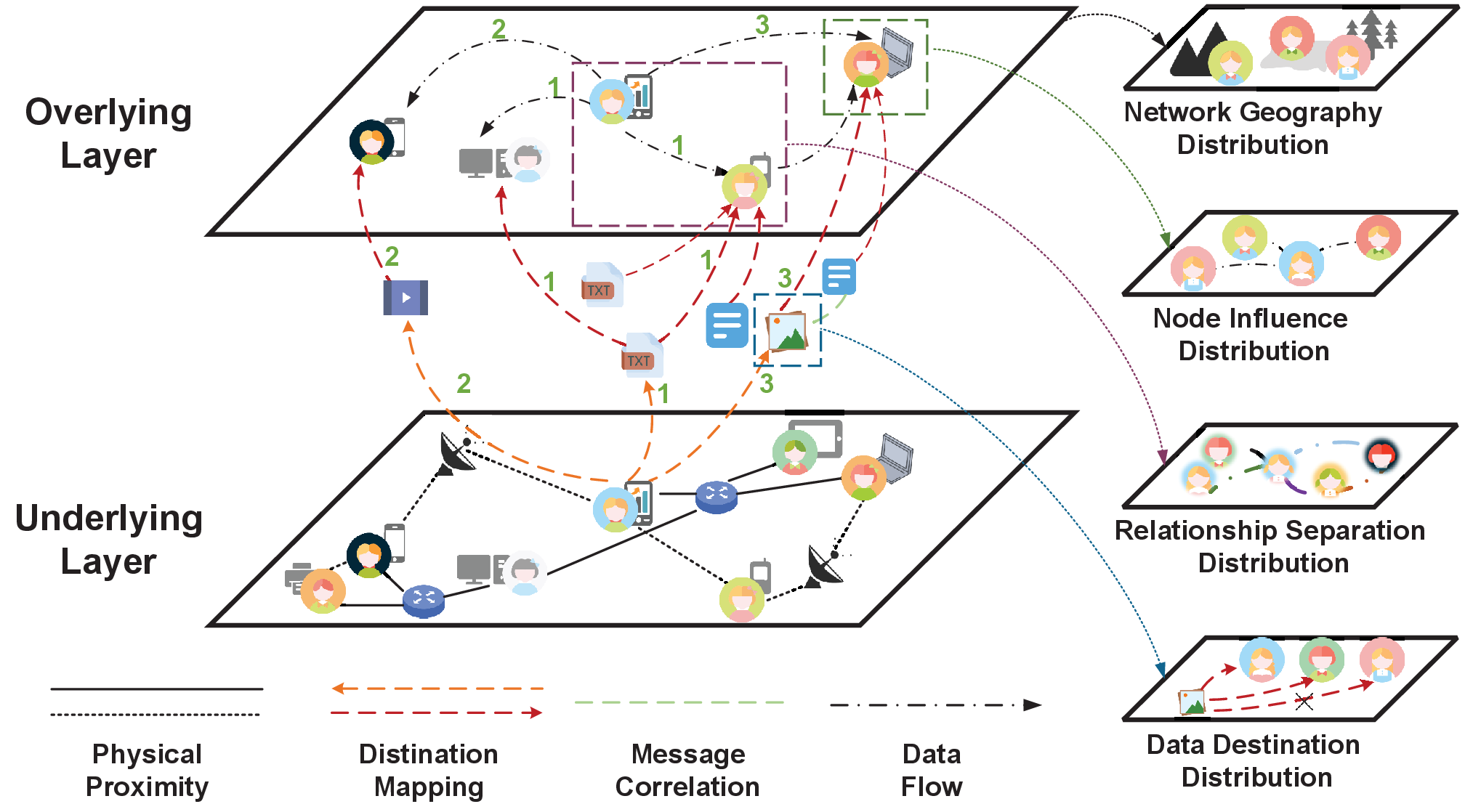}
			}
		\subfloat[The scope of different laws corresponding to the conditions while $\lambda=\Theta(1)$]{
				\includegraphics[width=0.9\linewidth]{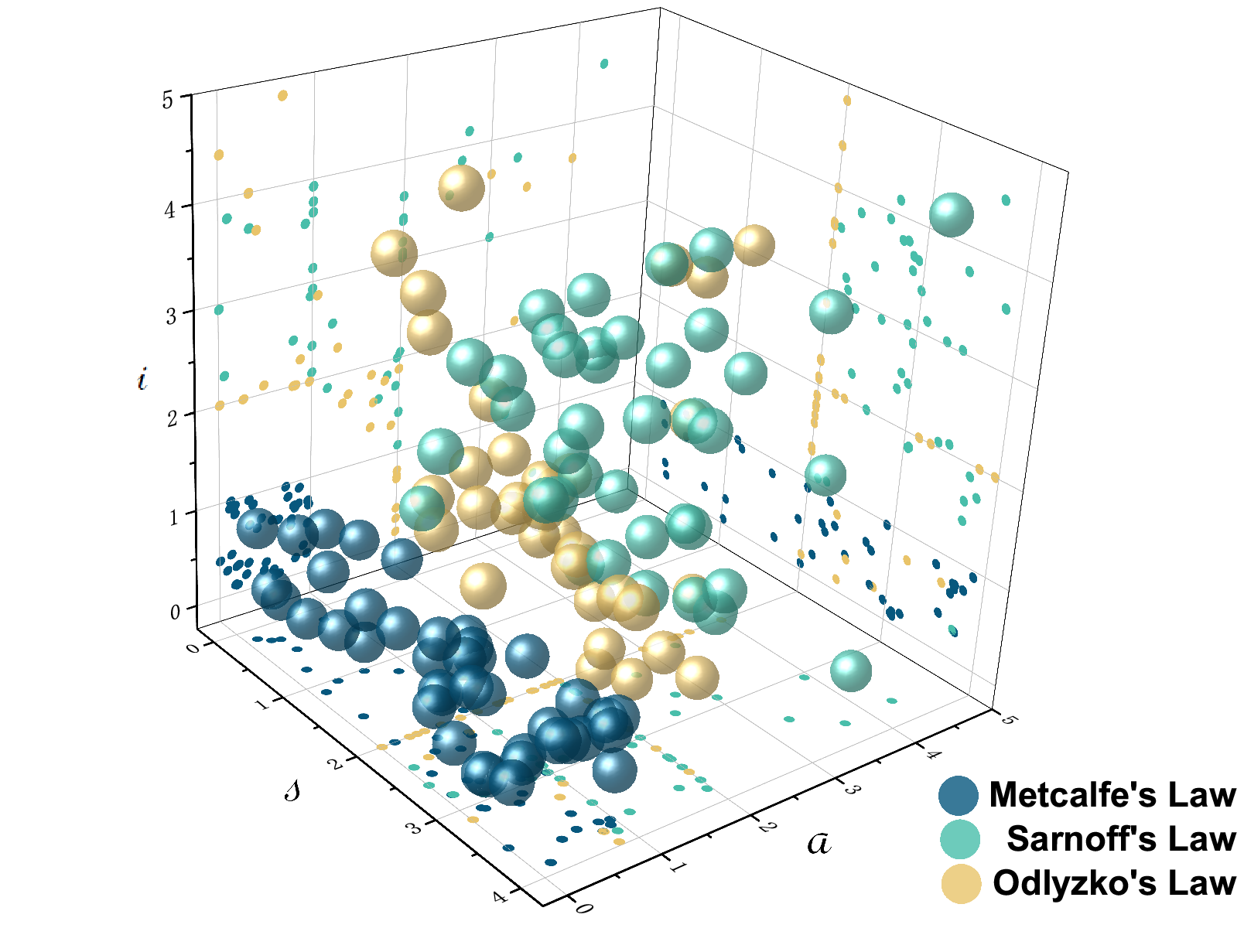}
			}
		\end{minipage}
	\caption{\textbf{The mechanisms of Metcalfe's Law.} \textbf{a} The proposed two-layer network model. The underlying layer handles physical data transmission, while the overlying layer facilitates device/user interactions. This work introduces four parameters to give understandings for the completeness of the overlying layer.
		\textbf{b} The scope of different laws corresponding to the conditions while $\lambda=\Theta(1)$. It can be noticed that the region associated with Metcalfe's Law is the most concentrated. We remark that under the condition $\lambda=\Theta(1)$, the cube law has no corresponding conditions.}
\end{figure}

\section*{Generation Mechanisms of Traffic Load}\vspace{-0.1in}
There are four generic aspects of network traffic load. 
To ensure completeness, we start by defining the overall structure of the network and its geographical distribution. Next, we consider various elements that constitute the network. Starting from the nodes, we examine the influence of different nodes. Taking into account the transmission choices of data generators and the interest selections of data receivers in networks, we define the data destination distribution for data transmissions within networks. From the perspective of relationships, we take into account the degree of closeness between relationships and define the separation of relationships. Detailed explanations are as follows.

We consider a network consisting of a random number $N$ (with $\bm{\mathbb{E}}[N]=n$) nodes distributed within a square region. The wraparound conditions are applied at the region edges by treating the network as the surface of a two-dimensional torus $\mathcal{O}$ to avoid border effects. For problem simplification, we make the assumption that the number of nodes is exactly $n$, and without impact on the result in order sense, denote the set of nodes $u_{k}$ by $\mathcal{U}$.

\subsubsection*{Network Geography Distribution}\vspace{-0.1in}
Considering networks covering regions with relatively uniform distribution densities. We assume the model to be a \textit{homogeneous random extended network} as \cite{15Franceschetti2007,grimmett1999percolation}.
The definition of network geography distribution can be extended to general network models, such as the clustering random model \cite{mobihoc2014} constructed the following procedure: First, making a ceter of $\mathcal{O}$ as the center point, denoted by $O$. Then, the center point $O$ generates a point process of nodes whose local intensity at position $X$ is given by $\mathbf{d}(Y)=n\cdot \kappa(O,X)$, where $\kappa(O,X)$ is a dispersion density function. Moreover, we assume that $\kappa(O,\cdot)$ is a summable, non-increasing, bounded and continuous function, and $\inf_{\mathcal{O}}\kappa(O,X)dX=1$. Following a common normalizing method,  $\kappa(O,\cdot)$ can be specified by first defining a non-increasing, bounded and continuous function $g(s)$ as $\kappa(O,X)= \frac{g(|X-O|)}{\int_{\mathcal{O}}g(|Y-O|)dY}$. Specifically, the clustering function is defined as $g(s):=\min\{1,s^{-\mathcal{g}}\}$, and $\mathcal{g}$ is the exponent to bound the function. It satisfies the assumption when $\mathcal{g}=0$ and $g(s)=1$.

\subsubsection*{Node Influence Distribution}\vspace{-0.1in}
Real-life scenarios often exhibit a pattern where only a limited number of nodes have high connectivity degrees, while the majority of nodes have relatively lower degrees. For instance, in online social networks, most people are primarily connected with those they know well, while a small subset of users such as celebrities, may have a large number of connections. 
In the assumption, the number of neighbouring nodes $q_k$ for a given node $v_k$ follows a Zipf's distribution \cite{manning1999foundations,sala2010brief}, i.e., $
\Pr(q_k=q)\propto {q^{-\mathcal{i}}}$, where $\mathcal{i}$ serves as an exponent in node influence distribution. Independent of the network system and the identity of its constituents, the probability $\Pr(q)$ that a vertex in the network interacts with $q$ other vertices decays as a power law.
A power-law distribution is a reasonable approach for capturing this characteristic, as it corresponds to the observed pattern of a few nodes exhibiting high degrees, while the majority of nodes show lower connectivity levels. There is a higher probability that it will be linked to a vertex that already has a large number of connections.
\subsubsection*{Data Destination Distribution}\vspace{-0.1in}
The distribution of destinations in network exhibits characteristics similar to Zipf's distribution \cite{manning1999foundations,sala2010brief}, denoted as $r_k$.
Mathematically,
$\mathrm{Pr}(r_{k}=r|q_{k}=q)\propto r^{-\mathcal{d}}, $
where $r_k$ represents the number of destinations for a data transmission session from node $v_k$, while $\mathcal{d}$ denotes the exponent of the data destination distribution. From the perspective of data generators, a typical example is the delivery network, where the destination for data delivery is predetermined at the moment of data generation, and the data generators are responsible for deciding who receives the data. From the perspective of data receivers, a typical example of this scenario is online social networks \cite{mobihoc2014}. Only users genuinely interested in specific content choose to accept and interact with the content shared by central users. The ultimate outcome will introduce a certain degree of randomness in the number of data destinations. This prompts the use of data destination distribution to synthesize these situations.

\subsubsection*{Relationship Separation Distribution}\vspace{-0.1in}
For the relationship separation distribution, let $\mathcal{D}(u, r)$ represent the disk with center at a point $u$ and radius $r$ within the deployment region $\mathcal{O}$, and define $N(u, r)$ as the count of nodes within $\mathcal{D}(u, r)$.
Kleinberg \cite{kleinberg2000navigation} proposed a distance-based social model relating geographical distance and social friendship, while Liben-Nowell et al. \cite{liben2005geographic} introduced the rank-based model. The rank-based model states that the friendship probability depends on both the geographic distance and node density. We use the population-based model \cite{mobihoc2014} by modifying the rank-based model. It is convenient to bound the total length of Euclidean spanning trees in the following derivation. Considering node $v_k$ as the reference point, this distribution can independently select $q_k$ points within the torus region $\mathcal{O}$ using a density function \cite{mobihoc2014}
\begin{equation}\label{equ-beta}
f_{v_k}(X) = \Phi_k\left(S, \mathcal{s}\right) \cdot \left ( \mathbb{E}\left[N(v_k, |X-v_k|) \right] +1 \right )^{-\mathcal{s}} ,
\end{equation}
where the position of a selected point in region $\mathcal{O}$ is denoted by the random variable $X$, and $|X-v_k|$ represents the Euclidean distance between $X$ and $v_k$. The relationship separation exponent is represented by $\mathcal{s}\in [0, \infty)$, and the coefficient $\Phi_k\left(S, \mathcal{s}\right)>0$ depends on both the area $S$ of region $\mathcal{O}$ and $\mathcal{s}$, satisfying
$
	\Phi_k\left(S, \mathcal{s}\right) \cdot \int_{\mathcal{O}} \left ( \mathbb{E}[N(v_k, |X-v_k|)] +1 \right )^{-\mathcal{s}} dX =1.
$

From the population-based model, $q_k$ anchor points are defined as the each points of the neighbour of $v_k$. That is, the $q_k$ friends of $v_k$ are located according to the positions of these corresponding anchor points. The set of anchor points $\mathcal{A}_k=\{p_{k_1},p_{k_2},\cdots,p_{k_{q_{k}}}\}$ is constructed by the following procedure. First, select arbitraily a point from $\mathcal{O}$, as reference point $v_k$. Second, select independently other $q_k$ points $p_{k_i}$ at random according to the density function as described in (\ref{equ-beta}).

Incorporating an exponent enhances the insight into node distribution across distances from a central node, considering the number of hops to reflect relationship proximity. Crucial in social networks where physical distance does not limit connections, and in telecom networks where operational range is limited, this approach helps assess service extent and understand spatial coverage and concentration of network resources by quantifying the clustering exponent.


Employing a model that integrates these four components, it is demonstrated that they are the cause of the power-law scaling of traffic load observed in actual networks. These components have a readily recognizable and significant role in the development of numerous networks, suggesting that the findings are pertinent to a broad spectrum of natural networks.

\section*{Derivation of Traffic Load}\vspace{-0.1in}

We next show that a model based on these four ingredients naturally leads to the scaling behavior of network traffic load, therefore connecting to Metcalfe's Law and its variants.

Starting with the number of node $p_k$ by applying the node influence distribution, the set of anchor points $\mathcal{A}_k$ is constructed from the relationship separation distribution. Since there are not all neighbour node can get the information, the number of transported destinations follow the data destination distribution. Based on $\mathcal{P}_k:=\{v_k\}\cup \mathcal{A}_k$, the Euclidean minimum spanning tree can be constructed, using some classic greedy algorithms like Prim algorithm. It provides a framework that minimizes the total distance between all points in the network while maintaining a connected structure.  

By integrating these methodologies including determining destinations, selecting anchor points, constructing an EMST, a comprehensive and nuanced understanding of the traffic load can be obtained.
It is observed that the traffic load aatisfies that 
$
\mathrm{L}_{\mathrm{N}}^{\ast}= \lambda \cdot\Omega\left( \sum_{k=1}^{n}|\EMST(\{u_{k}\}\cup \mathbb{D}_{k})|\right),
$
where $\lambda$ is the data arrival rate, $\EMST(\cdot)$ denotes the Euclidean minimum spanning tree over a set, and $\mathbb{D}_k$ represents all destinations of $u_k$. To compute this, we consider whether $q_k$ is with a constant level of influence $\Theta(1)$. When $q_k=\Theta(1)$, the Euclidean minimum spanning tree is hold as the aggregated distances for $u_k$ and anchor points set. When $q_k=\omega(1)$, we further consider the $r_k$. For $r_k=\Theta(1)$, the order of result is lower than the condition in $q_k=\Theta(1)$. For $r_k=\omega(1)$, the number of distinations is not a constant level, and we use the summation for all destinations.

Finally, we combine the network geography distribution by assuming the $\mathcal{g}=0$ as uniform distribution and relationship separation distribution by delving into the density function for nearest node set of each source node. From the relationship separation distribution, the total edge length of Euclidean minimum spanning tree is provided. The detailed scaling law of network traffic load is showed in Extended Data Table 2. 

%
%
%
%
%

\begin{figure*}[t]
	\centering
	\subfloat[Sarnoff's Function $y=an+b$]{\begin{minipage}{0.48\linewidth}
			\includegraphics[width=\textwidth]{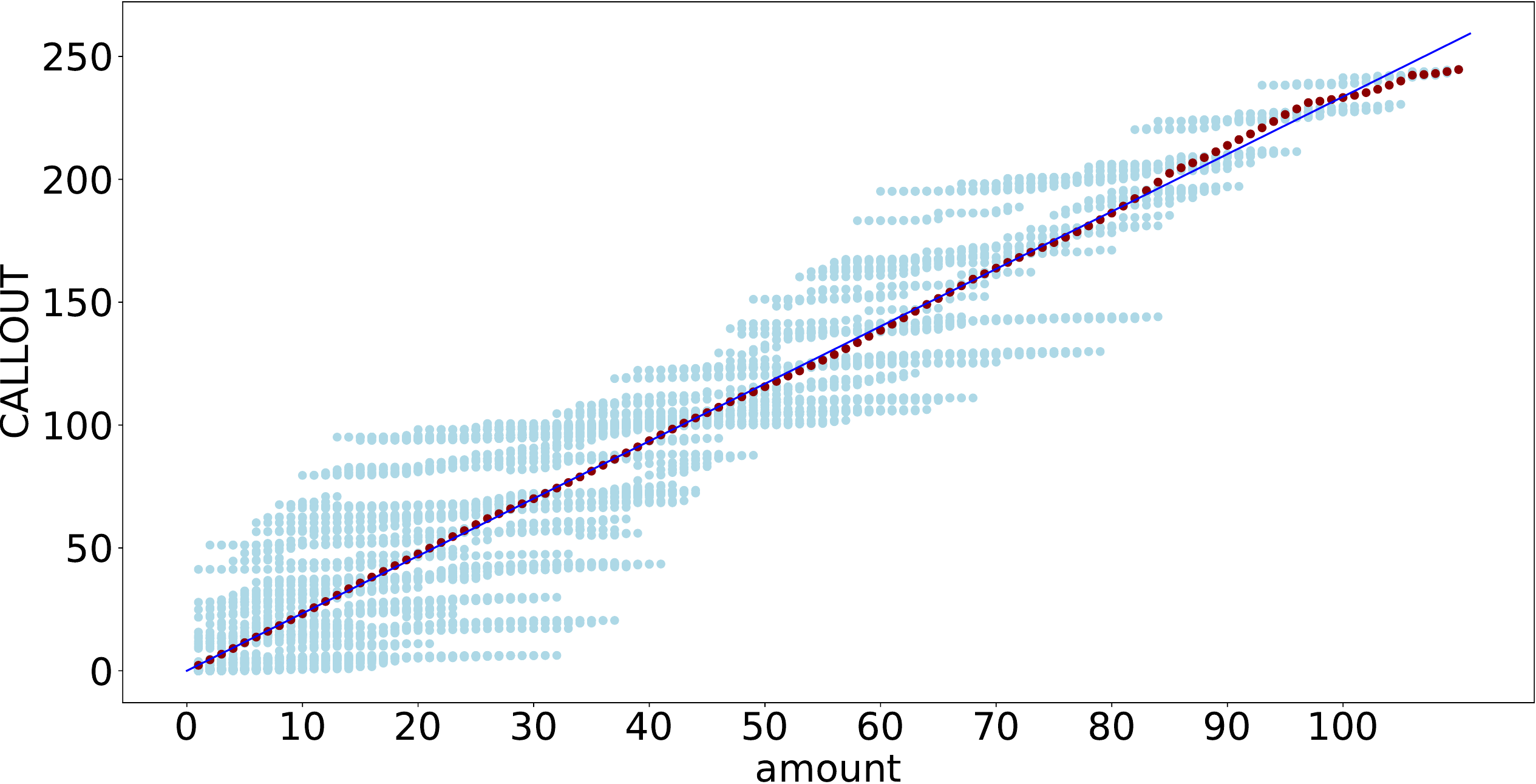}
	\end{minipage}}
\subfloat[Odlyzko's Function $y=an\ln n+bn +c$]{\begin{minipage}{0.48\linewidth}
		\includegraphics[width=\textwidth]{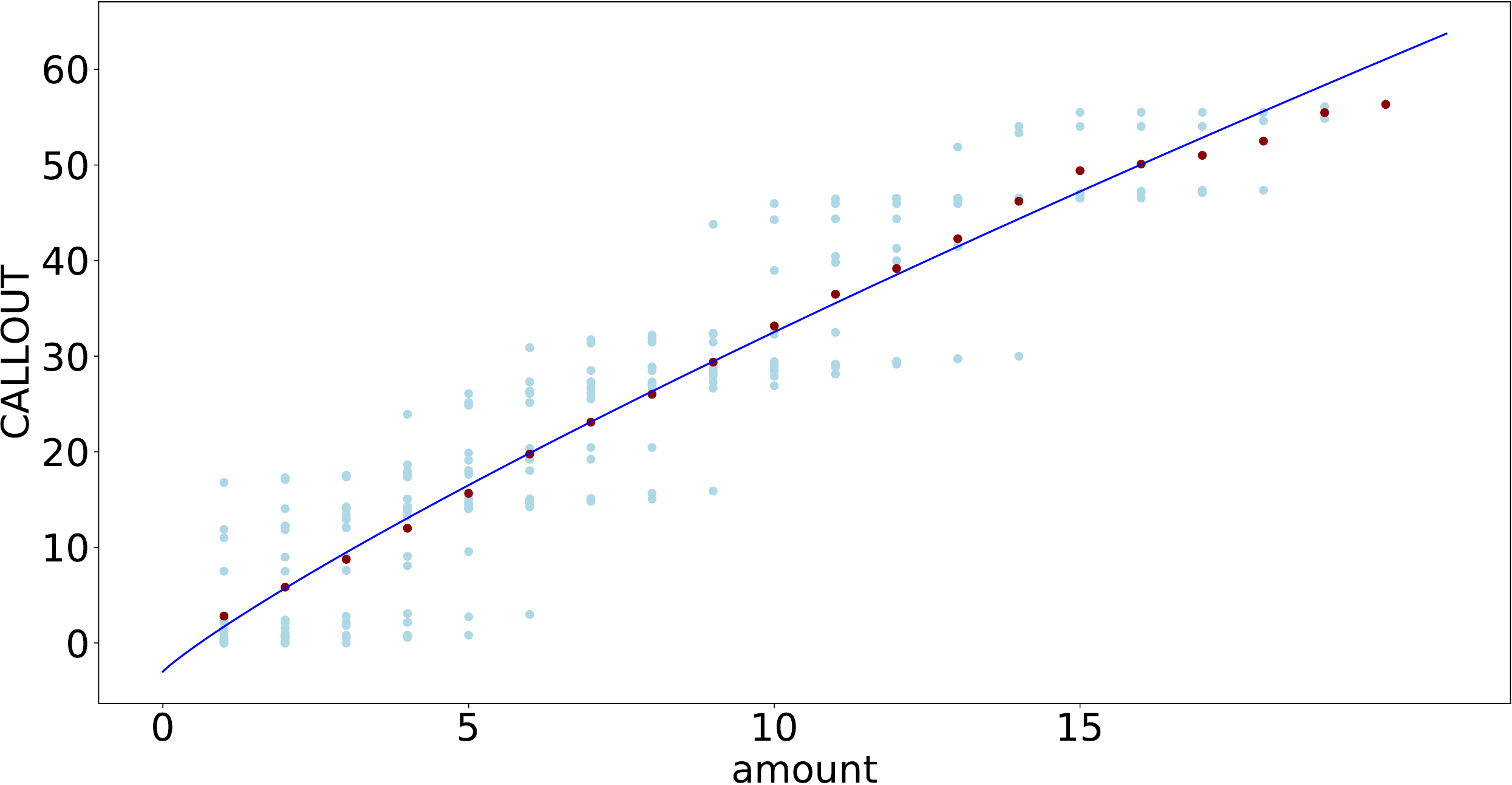}
\end{minipage}}

\subfloat[Metcalfe's Function $y=an^2+bn +c$]{\begin{minipage}{0.48\linewidth}
	\includegraphics[width=\textwidth]{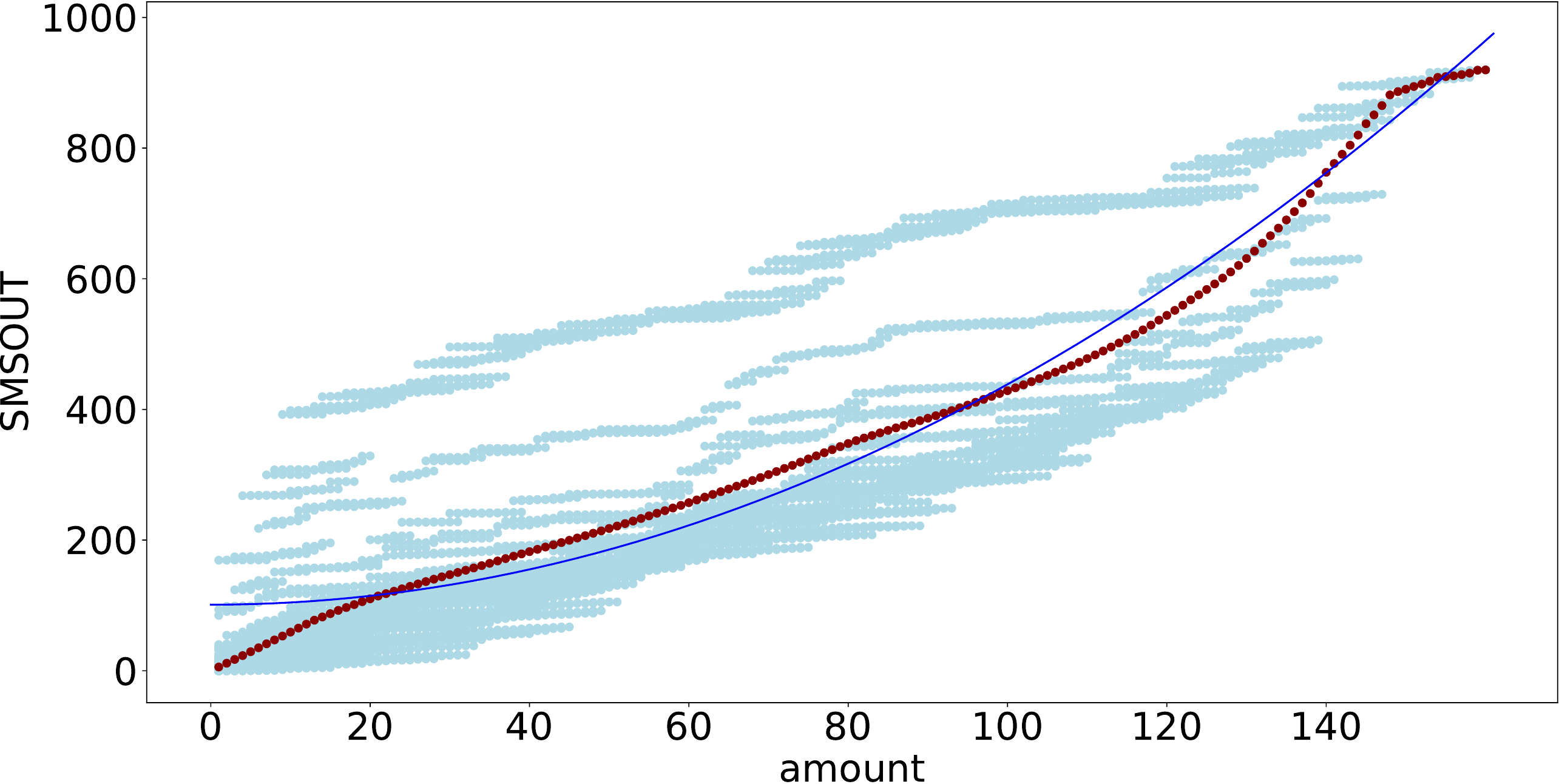}
\end{minipage}}
\subfloat[Cube Function $y=an^3+bn^2 +cn+d$]{\begin{minipage}{0.48\linewidth}
\includegraphics[width=\textwidth]{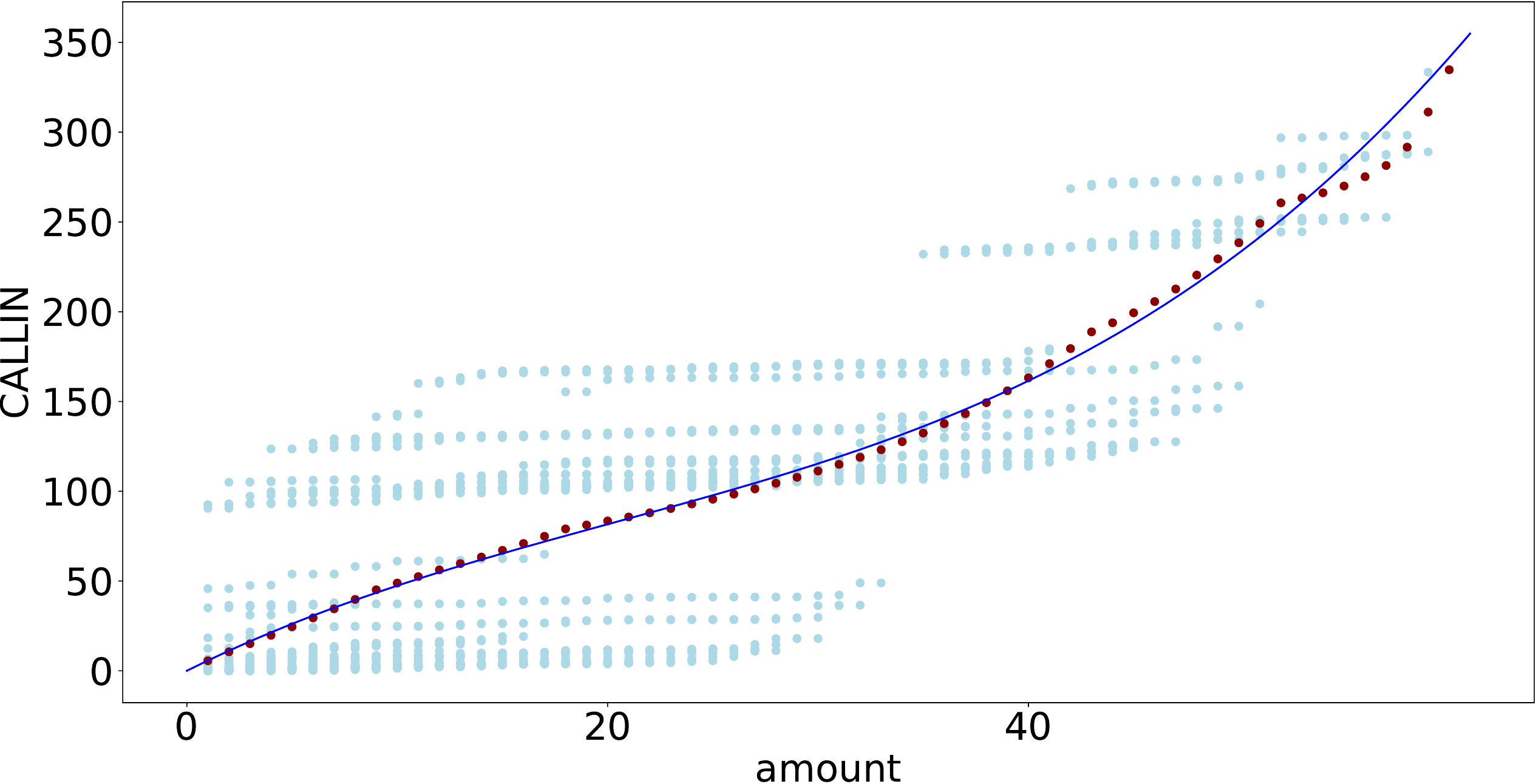}
\end{minipage}}
\caption{The empirical curve fitting results of Metcalfe's law and its variants, which can be seen as $y=an+b,y=an\ln n +bn+c, y=an^2+bn+c, y=an^3+bn^2+cn+d$, corresponding to Sarnoff's function, Odlyzko's function, Metcalfe's function, and Cube function, respectively.}
\label{pic-fit}
\vspace{-0.1in}
\end{figure*}

\section*{Reproducing Metcalfe's Law}
\vspace{-0.1in}

The above distribution are mechanisms common to a number of complex network, including social networks, transportation networks, and so on. Consequently, it is expected that Metcalfe's Law and its variants which have been available to us are the results of the combined action of the above mechanisms.

Some possible cases associated with Metcalfe's Law are summarized and examined in Extended Data Table \ref{tab-law-condition}. It illustrates the scenarios where Metcalfe's Law holds under different values of $\lambda$. Notably, when $\lambda=\Theta(1)$, we observe the consistent alignment with Metcalfe's Law for $\mathcal{s}\geq 0$, $0\leq\mathcal{d}<1$, and $0\leq \mathcal{i}<1$. This suggests that as long as the network has a positive exponent parameter $\mathcal{s}$ and a significant number of active users, Metcalfe's Law is applicable to describe its growth and value creation. A typical example is online social networks, such as Facebook and Tencent, which have extensive user bases and high user engagement.

Additionally, we explore the case of $\lambda=\Theta(n^{1/2})$ and indicate five distinct scenarios yielding similar conclusions. As a matter of fact, with the advancement of the Internet, it has become evident that the scale of networks is progressively expanding, and concurrently, the speed at which data is transmitted within these networks is also increasing. This growth is not just in terms of the number of users or devices connected to the network, but also in the infrastructure that supports the network, such as the routers, switches, and servers that facilitate the exchange of information. For $\mathcal{s}$ within $[0,1)$, the network is characterized as having a more open structure with a larger node base impacting the rate of data arrival. These findings indicate that Metcalfe's Law still holds in this state.

%

To validate the principles of Metcalfe's Law, Zhang et al. \cite{zhang2015tencent} conducted an empirical research study utilizing data from two prominent social network entities, Facebook and Tencent. Through the research, we emphasize the wide-ranging applicability of Metcalfe's Law, as it effectively captures network growth across diverse parameter spaces. The study showcases how the law's underlying principles hold in the context of various social networks, serving as a valuable tool for understanding and predicting network behavior and value.

Metcalfe's Law and its variations are verified using the dataset for real-world scenarios \cite{barlacchi2015multi}, revealing distinct functional relationships in traffic under varying node quantities and diverse scenarios. Conducting random selections and averaging for different node quantities, Metcalfe's Law and its variants are successfully fitted, as shown in Figure 2. By combining our results with previous findings \cite{zhang2023facebook, zhang2015tencent, briscoe2006metcalfe}, we assert the empirical validation of these laws.




\section*{Conclusion}\vspace{-0.1in}

This study primarily focuses on providing mechanisms and models behind Metcalfe's Law and its variants from the perspective of network traffic load. To address the limits of these laws, the potential theoretical conditions are established under which Metcalfe's Law and its variants hold and correlate these laws with real-world scenarios. The analysis reveals that these laws are applicable in distinct scenarios and have their unique significance. Particularly, we present the theoretical results across the complete parameter space for a general network model, offering valuable insights for exploring scenarios associated with the discovery of new scaling laws beyond the ones mentioned above in the future.

\newpage
\section{Methods}\label{Theo_Ana}

\subsection{Definition of Traffic Load}\label{Def-Tra}
\noindent\textbf{Data Arrival Rate.} Numerous studies have investigated the data transmission process within networks  \cite{karagiannis2004nonstationary,terdik2008levy,perera2010twitter,benevenuto2009characterizing} and we follow their discovery. The behavior of data arrival at a node follows a Poisson process.
We  denote the set of nodes by $\mathcal{U}=\{u_1,u_2,\cdots,u_n\}$. The rate vector is defined
\begin{equation*}
	\label{poisson rate}
	\bm{\lambda}=(\lambda_{1}, \lambda_{2}, \cdots, \lambda_{n}),
\end{equation*}
where $\lambda_{k}$ is the rate of the Poisson process for node $u_{k}$, $k=1, 2, \cdots, n$, called \textit{data arrival rate}.
For problem simplification, we assume that all $\lambda_{k}$ are in the same order, denoted as $\lambda$. It is reasonable that the data arrival rate would be the same order owing to the bandwidth and content delivery in the same network.


\noindent\textbf{Data Transport Distance.}
Denote a network $\mathrm{N}$ with a transportation scheme $\mathrm{S}$ as $\mathrm{S}_\mathrm{N}$, define a distance vector
\[\bm{d}_{\mathrm{S}_\mathrm{N}}=(d_{{\mathrm{S}_\mathrm{N}},u_1},d_{{\mathrm{S}_\mathrm{N}},u_2},\cdots,d_{{\mathrm{S}_\mathrm{N}},u_n}),\]
where
$d_{{\mathrm{S}_\mathrm{N}},u_k}$ indicates the distance data transporting from $u_{k}$ to destinations within a data transportation session $\mathrm{T}_{k}$. For the session $\mathrm{T}_{k}$, the source node $u_k$ delivers message to the selected neighbours.
The transport distance is influenced by the unique structures of communication networks and transmission methods.

\noindent\textbf{Traffic Load.}
Given a network $\mathrm{N}$, define the traffic load for $\mathrm{T}_{k}$, as
\[
\mathrm{L}_{\mathrm{N}}^{\mathrm{T}_k}=\lambda_k\cdot d_{{\mathrm{S}_\mathrm{N}},u_k}.\] 
As a result, the traffic load across all data transporatation schemes can be defined as
\begin{equation}\label{equ-transport-load}
	\mathrm{L}_{\mathrm{N}}^{\ast}= \min\nolimits_{\mathrm{T}_\mathrm{k}\in \mathbb{T}, \mathrm{S}_\mathrm{N}\in \mathbb{S}} \bm{\lambda}  \ast \bm{d}_{\mathrm{S}_\mathrm{N}},
\end{equation}
where $\mathbb{S}$ represents the set of all possible transmission schemes, $\mathbb{T}$ represents the set of all data transportation sessions regarding of $u_k$, and $\ast$ denotes an inner product.
It is shown that the traffic load is defined as the product of data arrival rate and data transport distance.

\subsection{Theoretical Analysis}
In this section, we delve into a detailed analysis of traffic load quantification in the network. We present a comprehensive model that allows us to quantify the traffic load and derive meaningful insights. To validate these findings, we provide a proof to support the conclusions.

\subsection{Basic Lemma to Derive the Traffic Load}

\begin{lem}[Minimal Spanning Tree \cite{steele1988growth}]\label{lem-growthsteel}
	Let $X_i$, $1\leq i <\infty$, denote independent random variables with
	values in $\mathbb{R}^d$, $d\geq 2$, and let $M_n$ denote the cost of
	a minimal spanning tree of a complete graph with vertex set $\{X_i\}_{i=1}^n$, where the cost of an edge $(X_i, X_j)$ is given by $\Psi((|X_i-X_j|))$.
	Here, $|X_i-X_j|$ denotes the Euclidean distance between $X_i$ and $X_j$ and
	$\Psi$ is a monotone function. For bounded random variables and $0<\sigma<d$,
	it holds that as $n\to \infty$, with probability $1$, one has
	\[
	M_n\sim c_1(\sigma,d)\cdot n^{\frac{d-\sigma}{d}} \cdot \int_{\mathbb{R}^d}f(X)^{\frac{d-\sigma}{d}}d X,\]
	provided $\Psi(x)\sim x^\sigma$, where $f(X)$ is the density of the absolutely continuous part of the distribution of the $\{X_i\}$.
\end{lem}

\begin{lem}[Kolmogorov's Strong LLN \cite{williams1991probability}]\label{lem-Kolmogorov-slln}
	Let $\{X_n\}$  be an i.i.d. sequence of random variables having finite mean: For $\forall n$,
	$\mathbb{E}[X_n] < \infty$.
	Then, a strong law of large numbers (LLN) applies to the sample mean:
	\[\bar{X}_n \stackrel{a.s.}{\longrightarrow} \mathbb{E}[X_n],\]
	where  $\stackrel{a.s.}{\longrightarrow}$ denotes \emph{almost sure convergence}.
\end{lem}

\begin{lem}\label{basic_lemma_emst} For $k=1,2,\cdots,n$, it holds that
	\begin{equation*}
		\begin{aligned}
			\mathrm{L}_{\mathrm{N}}^{\ast}= \lambda \cdot\Omega\left( \sum_{k=1}^{n}|\EMST(\{u_{k}\}\cup \mathbb{D}_{k})|\right),
		\end{aligned}
	\end{equation*}
	where $\EMST(\cdot)$ denotes the Euclidean minimum spanning tree over a set.
	
\end{lem}

See Lemma 1 of \cite{mobihoc2014} for the proof of this lemma.

\begin{lem}\label{lem-Pr}
	Consider the Zipf's distribution whose distribution function like $\Pr(q_k=q)={\left(\sum\nolimits_{j=1}^{n-1}j^{-\alpha}\right)^{-1}} \cdot {q^{-\alpha}}$, where $\alpha$ is a parameter to describe the distribution, we get that
	\begin{equation}\label{eq-Pr}
		\Pr(q_k=q)= \left\{
		\begin{array}{ll}
			\Theta\left( q^{-\alpha} \right)  , &  \alpha>1;  \\
			\Theta\left(\frac{1}{\log n}\cdot q^{-1}  \right), &  \alpha=1;  \\
			\Theta\left(n^{\alpha-1}\cdot q^{-\alpha} \right)  , & 0< \alpha< 1.
		\end{array}
		\right.
	\end{equation}
\end{lem}

\subsection{Euclidean Minimum Spanning Tree}
Denote a network session by an ordered pair $\mathrm{T}_k=<v_{k}, \mathcal{A}_k>$, where $ v_{k} $ is the source and each element $ v_{k_{i}} $ in $\mathcal{A}_k=\{v_{k_i}\}_{i=1}^{r_k}$ is the nearest node to the corresponding $ p_{k_{i}} $ in $ \mathcal{d}^\mathrm{I}_k=\{p_{k_i}\}_{i=1}^{r_k} $, the random variable $r_{k}$ denotes the number of potential destinations for session $\mathrm{T}_k$, i.e.,
The destination for node $v_k$ is dynamically selected based on the data being transported within the network during the session. 
We call point $p_{k_{i}}$ the \emph{anchor point} of $v_{k_{i}}$, and define a set $\mathcal{P}_{k} := \{v_{k}\} \cup \mathcal{A}_k$. Then, we can get the following lemma.

\begin{lem}\label{lem-EMST-B-D} For a network broadcast session $\mathrm{T}_k$, when $ r_{k}=\omega(1) $, with probability $1$, it holds that
	\begin{equation*}
		|\EMST(\mathcal{A}_k)|=\Theta(L_{\mathcal{P}}(\mathcal{s}, r_k)),
	\end{equation*}
	and then
	\begin{equation*}
		|\EMST(\mathcal{P}_{k})|=\Omega(L_{\mathcal{P}}(\mathcal{s}, r_k)),
	\end{equation*}
	where $\EMST(\cdot)$ denotes the Euclidean minimum spanning tree over a set,
	\begin{equation}\label{eq-L-P-beta-9765}
		L_{\mathcal{P}}(\mathcal{s}, r_k) = \left\{
		\begin{array}{ll}
			\Theta\left( \sqrt{r_k} \right)  , &  \mathcal{s}>2;  \\
			\Theta\left(\sqrt{r_k} \cdot \log n  \right), &  \mathcal{s}=2;  \\
			\Theta\left(\sqrt{r_k} \cdot n^{1-\frac{\mathcal{s}}{2}} \right)  , & 1< \mathcal{s} < 2;  \\
			\Theta\left(\sqrt{r_k} \cdot \sqrt{\frac{n}{\log n}} \right) , &  \mathcal{s}=1;  \\
			\Theta\left(\sqrt{r_k} \cdot \sqrt{n} \right)  , &  0\leq\mathcal{s} <1.
		\end{array}
		\right.
	\end{equation}
\end{lem}

\begin{proof}
	With probability $1$, it holds that
	\begin{center}
		$	|\EMST(\mathcal{A}_k)|=\Theta(L_{\mathcal{P}}(\mathcal{s}, r_k))$ for $r_k=\omega(1)$,
	\end{center}
	where $L_{\mathcal{P}}(\mathcal{s}, r_k)$ is defined in (\ref{eq-L-P-beta-9765}). Combining with the fact that \[|\EMST(\mathcal{P}_{k})|\geq |\EMST(\mathcal{A}_k)|,\] then, we get that
	\begin{equation*}
		|\EMST(\mathcal{P}_{k})|=\Omega(L_{\mathcal{P}}(\mathcal{s}, r_k)).
	\end{equation*}
	Then, our focus turns to the derivation of $ L_{\mathcal{P}}(\mathcal{s}, r_k) $. The  given $X_0$ is a point from $\mathcal{O}$ selected arbitrarily and make $X_0$ as the reference point. The distribution of points in $\mathcal{A}_k$ follows the density function
	\begin{equation*}
		f_{X_0}(X)=\frac{\left(\int_{\mathcal{D}(X_0,|X-X_0|)}\mathbf{d}(Y)dY+1\right)^{-s}}{\int_{\mathcal{O}}\left(\int_{\mathcal{D}(X_0,|X-X_0|)}\mathbf{d}(Y)dY+1\right)^{-s}dZ},
	\end{equation*}
	where 
	$\mathbf{d}(Y)=n\sum\nolimits_{c_j\in \mathcal{C}} \frac{g(|Y-c_j|)}{\int_{\mathcal{O}}g(|Z-c_j|)dZ}$ and $\mathcal{C}$ is the positions set of points from \cite{mobihoc2014}. 
	
	While we assume that $\mathcal{g}=0$, we get $g(\cdot)=1$ and $\mathbf{d}(Y)=1$. Combining the conclusion of Lemma \ref{lem-Pr}, the density function can be calculated as:
	\begin{equation*}\label{eq-g-f-X0-X}
		f_{X_0}(X) = \left\{
		\begin{array}{ll}
			\Theta\left( (|X-X_0|^2+1)^{-\mathcal{s}} \right)  , &  \mathcal{s}>1;  \\
			\Theta\left(\frac{1}{\log n}\cdot (|X-X_0|^2+1)^{-1}  \right), &  \mathcal{s}=1;  \\
			\Theta\left(n^{\mathcal{s}-1}\cdot(|X-X_0|^2+1)^{-\mathcal{s}} \right)  , & 0< \mathcal{s} < 1.
		\end{array}
		\right.
	\end{equation*}
	
	Based on the value of $\mathcal{s}$, we have:

	%
	%
	%
	%
	%
	%
	%
	%
	
	(1) When $ \mathcal{s}>1 $,
	\begin{equation*}
		\int_{\mathcal{O}}\sqrt{f_{X_0}(X)}dX = \Theta\left( \int_{\mathcal{O}}\frac{dX}{\left(|X-X_{0}|^{2}+1\right)^{\frac{\mathcal{s}}{2}}}\right)
	\end{equation*}
	\begin{equation*}
		~~~~~~~~~~~~~~~~~~~~~~=\left\{
		\begin{array}{ll}
			\Theta\left(1\right) , & \mathcal{s}>2;\\
			\Theta\left(\log n\right) , & \mathcal{s}=2;\\
			\Theta\left(n^{1-\frac{\mathcal{s}}{2}}\right) , & 1<\mathcal{s}<2.
		\end{array}
		\right.
	\end{equation*}

	(2) When $ \mathcal{s}=1 $,
	\begin{equation*}
		\begin{aligned}
			\int_{\mathcal{O}}\sqrt{f_{X_0}(X)}dX &= \Theta\left(\frac{1}{\sqrt{\log n}}\cdot \int_{\mathcal{O}}\frac{dX}{\left(|X-X_{0}|^{2}+1\right)^{\frac{1}{2}}}\right)\\
			& = \Theta\left(\sqrt{\frac{n}{\log n}}\right).
		\end{aligned}
	\end{equation*}
	
	(3) When $ 0 \leq\mathcal{s} <1 $,
	\begin{equation*}
		\begin{aligned}
			\int_{\mathcal{O}}\sqrt{f_{X_0}(X)} dX &=  \Theta\left(n^{\frac{\mathcal{s}-1}{2}} \cdot \int_{\mathcal{O}}\frac{dX}{\left(|X-X_{0}|^{2}+1\right)^{\frac{\mathcal{s}}{2}}}\right)\\
			&= \Theta\left(\sqrt{n}\right).
		\end{aligned}
	\end{equation*}
	
	Combining with all cases above, we complete the proof.$\hfill\square$
\end{proof}

\subsection{Main Results on Traffic Load}\label{Z-D-result-interestcast}
The following theorem demonstrates a theoretical bound on the traffic load for network $\mathrm{N}$ as follows.
\begin{thm}\label{Z-D-thm-traffic-load-on-sessions}
	Let $\mathrm{L}_{\mathrm{N}}^{\ast}$ denote the traffic load for data transportation in network $\mathrm{N}$ with $n$ nodes.
	It holds that
	\begin{equation*}
		\mathrm{L}_{\mathrm{N}}^{\ast}=\mathrm{L}_{\mathrm{N}}^{\ast}\left(\lambda,\mathcal{i},\mathcal{s},\mathcal{d},n\right),
	\end{equation*}
	where the value of $\mathrm{L}_{\mathrm{N}}$, presented in Extended Data Table \ref{G-tab-total},
	depends on  the data arrival/generating rate $\lambda$, exponent of node influence distribution $\mathcal{i}$, exponent of relationship separation distribution $\mathcal{s}$ and exponent of data destination distribution $\mathcal{d}$.
	
\end{thm}


\subsection{Proof of Theorem \ref{Z-D-thm-traffic-load-on-sessions}}
\begin{proof}	
	We present at the logical framework for the proof. For conciseness, we refer to several common conclusions from \cite{mobihoc2014, wang2016modeling}. The scenario they studied, i.e., online social networks, is indeed a special case for this work. There are still some commonalities in technical logic.
	
	In the context of the underlying network, the problem of achieving optimal transmission distance involves connecting the nodes in the given plane, which constitutes the set $|{u_k}\cup\mathbb{D}_k|$, where $\mathbb{D}_k$ represents all destinations of $u_k$. 
	We can employ the total length of the Euclidean spanning tree of a session to gauge the magnitude of the optimal transport distance.
	This problem can be simplified as generating the Euclidean Steiner tree for the set ${u_k}\cup\mathbb{D}_k$ \cite{gilbert1968steiner}.
	%
	%
	Then, we can get that
	\begin{equation*}
		\begin{aligned}
			\mathrm{L}_{\mathrm{N}}^{\ast}= \lambda \cdot\Omega\left( \sum_{k=1}^{n}|\EMST(\{u_{k}\}\cup \mathbb{D}_{k})|\right),
		\end{aligned}
	\end{equation*}
	where $\EMST(\cdot)$ denotes the Euclidean minimum spanning tree \cite{penrose1999strong}.
	
	Further, 
	let $\mathcal{A}_k=\{p_{k_1},p_{k_2},\cdots,p_{k_{q_{k}}}\}$ and $\mathcal{F}_k=\{v_{k_1},v_{k_2},\cdots,v_{k_{q_{k}}}\}$ respectively denote the set of  $q_k$ nodes and	the closest nodes to  $p_{k_i}$, for $1 \leq i \leq q_k$.
	Additionally, we denote a set
	$\mathcal{P}_k= \{v_k\} \cup \mathcal{A}_k$, where $v_k$ is the source from a given session $\mathrm{T}_k$ and call the anchor point of $v_{k_i}$ as $p_{k_i}$.
	Then, for a data transportation session $\mathrm{T}_k$, we can get 
	
	\vspace{-0.11in}
	
	{\small
		\begin{equation*}
			|\EMST(\mathcal{A}_k)|=\Theta(L_{\mathcal{P}}(\mathcal{s}, r_k)), 
			|\EMST(\mathcal{P}_{k})|=\Omega(L_{\mathcal{P}}(\mathcal{s}, r_k)),
		\end{equation*}
	}
	
	\vspace{-0.02in}
	
	\noindent where $L_{\mathcal{P}}(\mathcal{s}, r_k)=\sqrt{r_k}\cdot \int_{\mathcal{O}}\sqrt{f_{X_0}(X)}dX$. 
	In Lemma \ref{lem-EMST-B-D}, the result of $L_{\mathcal{P}}(\mathcal{s}, r_k)$ is in (\ref{eq-L-P-beta-9765}).
	\begin{equation}
		L_{\mathcal{P}}(\mathcal{s}, r_k) = \left\{
		\begin{array}{ll}
			\Theta\left( \sqrt{r_k} \right)  , &  \mathcal{s}>2;  \\
			\Theta\left(\sqrt{r_k} \cdot \log n  \right), &  \mathcal{s}=2;  \\
			\Theta\left(\sqrt{r_k} \cdot n^{1-\frac{\mathcal{s}}{2}} \right)  , & 1< \mathcal{s} < 2;  \\
			\Theta\left(\sqrt{r_k} \cdot \sqrt{\frac{n}{\log n}} \right) , &  \mathcal{s}=1;  \\
			\Theta\left(\sqrt{r_k} \cdot \sqrt{n} \right)  , &  0\leq\mathcal{s} <1.
		\end{array}
		\right.
	\end{equation}
	
	\par First, we consider the \textit{node influence distribution} and we define two sets for all data transportation sessions, where $\mathcal{K}^{1}$ represents the nodes with a constant level of influence $\Theta(1)$, and $\mathcal{K}^{\infty}$ denotes the nodes with a non-constant level of influence. Together, they fully characterize $\sum\nolimits_{k=1}^{n}|\EMST(\mathcal{P}_k)|$. In this proof, ${\Psi}^{1}$ is defined to represent the sum of the lengths of all the Euclidean minimum spanning trees under the condition $k \in \mathcal{K}^{1}$, where ${\Psi}^{1}=\sum_{k \in \mathcal{K}^{1}}|\EMST(\mathcal{P}_k)|$, and ${\Psi}^{\infty}$ is defined to represent the one under the condition $k \in \mathcal{K}^{\infty}$, where ${\Psi}^{\infty}=\sum_{k \in \mathcal{K}^{\infty}}|\EMST(\mathcal{P}_k)|$.

	We first address ${\Psi}^{1}$. Since for $q_k=\Theta(1)$, it holds that 	$|\EMST(\mathcal{P}_{k})|=\Theta(|X-v_k|)$. We define a sequence of random variables $\varepsilon_k:=|X-v_k|/\sqrt{n}$, and $\mathbb{E}[\varepsilon_k]=\mathbb{E}[|X-v_k|]/\sqrt{n}$. The value of $\mathbb{E}[|X-v_k|]$ is given in Lemma 5 of \cite{mobihoc2014} and ${\Psi}^{1}=\Theta(\sqrt{n}\sum_{k\in \mathcal{K}^{1}}\varepsilon_k )$. Through Kolmogorov's Strong laws of larger numbers (LLM) \cite{williams1991probability}, with probability 1,
	\begin{equation*}
		{\Psi}^{1}=\Theta\left(|\mathcal{K}^{1}| \cdot E[|X-v_{k}|] \right).
	\end{equation*}

	Next, the \textit{data destination distribution} is utilized to consider ${\Psi}^{\infty} $. All $ |\EMST(\mathcal{P}_k)| $ are independent, where $k \in \mathcal{K}^{\infty}$, due to the introduction of anchor points. Similarly to the node influence distribution, we apply the same approach to distinguish the number of nodes selected by this distribution. Specifically, we denote ${\Psi}_{1}^{\infty}$ and ${\Psi}_{\infty}^{\infty}$ as the sum of lengths of all the Euclidean minimum spanning trees under $r_k=\Theta(1)$ and $r_k=\omega(1)$, respectively.

	For $r_k=\Theta\left(1\right)$, the value of $ {\Psi}_{1}^{\infty} $ is relatively infinitesimal compared with ${\Psi}^{1} $.
	
	Then we consider $ \mathcal{K}^{\infty}_{\infty} $. We define $T_q=n\cdot \Pr(q_k=q)=n\cdot \left(\sum_{j=1}^{n-1} j^{-\mathcal{i}}\right)^{-1}\cdot q^{-\mathcal{i}}$, as the number of nodes with $q$ destinations.
	Through LLN \cite{williams1991probability}, with probability $1$, it holds that
	\begin{equation*}\label{Transport-for-d-wuqiong}
		{\Psi}_{\infty}^{\infty} \geq \sum_{q=2}^{n-1} \sum_{r=1}^{q} T_{q} \cdot \mathrm{Pr}(r_{k}=r|q_{k}=q) \cdot L_{\mathcal{P}}\left(\mathcal{s},r\right).
	\end{equation*}
	Actually, what we should calculate is $\sum_{q=2}^{n-1}\sum_{r=1}^{q}\Pr(q_k=q)\cdot \Pr(r_k=r|q_k=q)$, denoted as $G(\mathcal{i},\mathcal{d})$, which is illuastrated in Extended Data Table \ref{tab-G}.
	
	Combining with the above conclusions, the lower bounds are get on $\sum\nolimits_{k=1}^{n}|\mathrm{EMST}(\mathcal{P}_k)|$.

	Then we calculate $\sum\nolimits_{k=1}^n| \mathrm{EMST}(\mathbb{D}_k) |$. Since
	\begin{equation*}
		\sum_{k=1}^{n} r_{k} = \Theta\left(\sum_{q=1}^{n-1} \sum_{r=1}^{q} T_{q} \cdot \mathrm{Pr}(r_{k}=r|q_{k}=q) \cdot r \right),
	\end{equation*}
	where $T_{q}  =  n\cdot\mathrm{Pr}(q_{k}=q) $. We get the result of $ \sum\nolimits_{k=1}^{n} r_{k}$ as $W\left(\mathcal{i}, \mathcal{d}\right)$, as illustrated in Extended Data Table \ref{tab-W-gamma-varphi}.
	
	Additionally, for all $u_k\in \mathcal{U}$, 
	\begin{equation*}
		E[| v_{k_{i}} - p_{k_{i}} |] =\Theta\left(\int_0^{\sqrt{n}}x\cdot e^{-\pi x^2}dx\right).
	\end{equation*}
	
	That is,
	\begin{equation*}
		E[| v_{k_{i}} - p_{k_{i}} |] = \Theta \left(1\right).
	\end{equation*}
	and LLN, with probability $1$, it holds that
	\begin{equation*}\label{Length-Anchor-Friends}
		\sum_{k=1}^{n} \sum_{i=1}^{r_{k}} |v_{k_{i}}-p_{k_{i}}| = \Theta\left(\sum_{k=1}^n r_k\right)= \Theta\left(W\left(\mathcal{i}, \mathcal{d}\right)\right).
	\end{equation*}
	Combining the result of $\sum\nolimits_{k=1}^{n}|\mathrm{EMST}(\mathcal{P}_k)|$
	and the data arrival rate $\lambda$, we complete the proof.$\hfill\square$
	
\end{proof}

\subsection{How to Understand The Results}
Visual results are conditioned on multiple distribution indices, represented in Extended Data Figure \ref{pic-understand}. We have come to observe that the theoretical depiction of traffic load manifests breakpoints regardless of the conditions, signifying potential "phase transitions" in real-world scenarios due to varying conditions. These breakpoints play a pivotal role in delineating distinct shifts or transformations within the network dynamics. This observation parallels the existence of multiple variants of Metcalfe's Law, suggesting the presence of diverse breakpoints and transitions, thereby underlining the crucial role these breakpoints hold.

\subsection{Mechanisms to Reproduce Scaling Property}\label{Variants}

In this section, we offer an in-depth explanation of variants of Metcalfe's Law. Based on theoretical findings, we establish connections between these different families of laws, identifying some theoretical exemplary cases and real-world scenarios in which they are applicable. 

Here, we provide a concise introduction to the recurrence for the Metcalfe's Law family. Considering the influence of increasing the number of network nodes on the data arrival rate $\lambda$, we assume that $\lambda$ is a non-decreasing function of $n$. In the analysis, we explore different scaling behaviors for $\lambda$, specifically $\Theta(1)$, $\Theta(n^{1/2})$, and $\Theta(n)$. Based on the theoretical findings, we match different laws with their respective exemplary possible cases. We use these cases to explain the real-world implications of these laws and provide insights into Metcalfe's Law and its variants. Actually, from Extended Data Table \ref{tab-W-gamma-varphi}, we can also derive other laws corresponding to different cases. To demonstrate the importance of Metcalfe's Law, we consider the conditions of several other laws corresponding to different parameter settings while keeping $\lambda=\Theta(1)$ and $\mathcal{g}=0$, as shown in Figure 1(b).

\subsubsection{Sarnoff's Law: $n$}

\textbf{Theoretical Cases of Sarnoff's Law.}
%
%
These scenarios for Sarnoff's Law are characterized by strong regional correlations and low-degree distributions.
In networks characterized by limited openness, interactions between nodes are more constrained, especially when node activity levels are relatively low. Consequently, the amount of data transmitted between nodes within a given time period is significantly reduced. It is reasonable to use a constant level of $\lambda$ in such scenarios.

In such networks, the information flows between nodes are hindered due to the restricted nature of node interactions, which may arise from factors such as access permissions of the device. As a result, the level of data transmission within the network remains minimal, with limited exchange of information among nodes.

\noindent \textbf{Practical Analysis on Sarnoff's Law.}
The conditions for Sarnoff's Law primarily apply to networks with strong localization and a large but limited number of active users, leading to minimal node overlap. Originally formulated for broadcast media platforms, such as TV and similar channels, these networks have limited user interaction, focusing on one-way information transmission from broadcasting nodes to viewers. Hence, the network traffic and value depend on the number of nodes linearly.

In broadcast media networks, the value of the network is closely tied to the quantity of nodes, rather than the level of user interaction or engagement.  The main focus is on the dissemination of content from a central source to a large audience, with little to no direct interaction between viewers. The value of these networks are primarily measured by the reach and size of their audience.

\subsubsection{Odlyzko's Law: $n\log n$}
\textbf{Theoretical Cases of Odlyzko's Law.}
These conditions that identified the Odlyzko's Law indicate a high level of relationship under a  limited region, a large number of potential nodes, and relatively fewer destinations for nodes to select.

This set of conditions is associated with networks that exhibit strong localization, where the majority of nodes have limited connections and there are fewer active nodes. In such networks, the value of the network primarily stems from the potential user base rather than active user engagement. Odlyzko's Law challenges Metcalfe's Law by emphasizing that the value growth of communication networks is not solely dependent on the number of active connections but also considers the potential reach of the network.

Note that the applicability of Odlyzko's Law is contingent upon the specific conditions outlined above. These conditions highlight the characteristics of networks where the value growth is primarily driven by the potential user base rather than the level of active engagement.

\noindent\textbf{Practical Analysis on Odlyzko's Law.}
One typical example is carrier networks \cite{briscoe2006metcalfe}. In the carrier networks, it is a common practice for large ISPs (Internet Service Providers) to decline peering arrangements with smaller operators. These networks demonstrate a pronounced localization, a substantial user base potential, and relatively low user activity levels. In this context, the primary emphasis lies on interconnection between operators rather than extensive user-to-user interconnection. Consequently, the overall number of actively participating nodes within the network remains comparatively limited.

Different from the conditions outlined in Sarnoff's Law, networks guided by Odlyzko's Law exhibit a slightly elevated level of interactivity.
In general, user-to-user interconnection within the network is limited or almost non-existent. As a result, the network's predominant architecture centers around facilitating interconnection between operators rather than facilitating direct interactions between users.
In such a context, the number of actively engaged nodes remains constrained.

When taking into account these attributes, it is reasonable to introduce a factor of $\log n$, to capture the slightly heightened interactivity relative to Sarnoff's Law. This factor acknowledges the logarithmic growth in network effects stemming from the constrained interconnection among users within these networks \cite{briscoe2006metcalfe}.


%

\subsubsection{Cube Law: $n^3$}
\textbf{Theoretical Cases of Cube Law.}
As the number of users increases, the amount of incoming information within a unit of time may approach the total number of users, as demonstrated by the data from Zhang et al. \cite{zhang2023facebook} in 2020.

This trend emphasizes the importance of considering the growing user count when analyzing the information flow in networks. With an expanding user base, the network becomes more vibrant, fostering increased interactions and information sharing among users. The continuous growth in user numbers amplifies the rate of information arrivals, underscoring the notion that user interactions become increasingly intertwined as the user base expands.

\noindent\textbf{Practical Analysis on Cube Law.}
There are several factors that contribute to the observed phenomenon. Firstly, a larger user base in the network leads to a higher number of active participants. As more users engage in generating and sharing information, the overall rate of information arrivals naturally increases. This influx of data contributes to a dynamic and vibrant network ecosystem.


Additionally, with the increase in the number of nodes in the network, we observe a shift in the relationship between network value and network traffic. The traditional linear correlation is surpassed as the network scales. Real-world scenarios, particularly evident in online social networks, exemplify this trend. Each new user joining the network brings a wealth of additional content, connections, and interactions, all contributing to an enhanced user experience. This positive user experience attracts even more users to join the network, resulting in a self-reinforcing cycle of growth.

Furthermore, economies of scale play a significant role in this phenomenon. As the number of users increases, the network's infrastructure and resources become more efficient and cost-effective. This, in turn, enables the network to handle higher traffic loads and accommodates the growing user base without compromising on performance.

The cube law suggests a highly interconnected network, where users actively exchange information and contribute to the overall network value. This further reinforces the notion that a larger user base brings about greater economies of scale, driving network growth and efficiency.

\subsection{Empirical Evaluations}\label{Evaluation}

In this section, we provide an evaluation of the parameters in the theoretical model by utilizing several relevant datasets.
While prior work has extensively explored data validation of Metcalfe's Law and its variants, these validations have been conducted separately in various network contexts. 



\subsubsection{Network Geography Distribution}

In this section, we provide the evaluations of the network geography distribution through the Gowalla dataset. The Gowalla dataset is a comprehensive collection of location-based information obtained from the Gowalla social networking platform. It predominantly emphasizes geospatial data associated with check-ins and user activities across diverse physical locations.  Each check-in record generally encompasses details such as geographical coordinates (latitude and longitude) of the location, a timestamp, and occasionally supplementary metadata like the place's name or venue category.
This dataset presents a diverse and extensive assortment of geographic information, facilitating analyses concerning user mobility patterns, location-based recommendations, and insights into human behavior linked with spatial interactions.

In Extended Data Figure \ref{fig:final_map}, we present the overall geographical distribution of Gowalla users in the United States, utilizing a Mercator projection coordinate system. Building upon this, we select two specific locations to examine the distribution of geographical positions, as illustrated in Extended Data Figure \ref{fig_gowalla}, providing further clarification for the configurations.


In the evaluations, we specifically extract data from 29,105 users in Bloomington, Illinois (within the coordinates 85\degree W - 92\degree W and 37\degree N - 45\degree N, Extended Data Figure \ref{fig_gowalla}(a)), and 84,222 users in San Francisco, California (within the coordinates 121\degree W - 122\degree W and 37\degree N - 38\degree N Extended Data Figure \ref{fig_gowalla}(b)). Subsequent to a detailed analysis of the dataset, a noticeable trend emerges: the geographical distribution within real-world networks tends to favor uniformity. This observation robustly supports the credibility of the underlying assumption.

\subsubsection{Traffic Load}
The Telecommunication datasets \cite{barlacchi2015multi} provide data about the telecommunication activity in the city of Milan and in the Province of Trentino. Specifically, one for telecommunication activities and two for telecommunication interactions. The dataset serves as measure of the level of interaction between the users and the mobile phone network.

From the conclusions drawn in Figure 2, Extended Data Figure \ref{pic-fit2} and, Extended Data Table \ref{tab-law-fit}, we can analyze distinct outcomes. Notably, when fitting a cubic function, the coefficient of the cubic term is significantly smaller than those of the quadratic and linear terms. However, when employing alternative functions for fitting, their performance does not match that of the cubic function. Hence, even though the coefficient of the cubic term is notably small, its impact on the function remains substantial. Consequently, when real-world scenarios exhibit changes in concavity, the cubic function demonstrates sufficient complexity to capture these variations effectively.

\subsection{Discussion}
To further indicate the connection between the work and the Metcalfe's Law family, this section delves into the universality of the proposed network framework and explores the relationship between the traffic load and value of networks.


$\bullet$ In the model, each individual node within the network possesses the capability to not only act as a data receiver but also as a data source. This assumption is supported by empirical evidence in real-life scenarios. Across the spectrum of communication technologies, ranging from the deployment of 4G in mobile internet to the revolutionary implementation of 5G in the Internet of Things (IoT), this dual functionality of network nodes has consistently manifested itself. The introduction of this generic model provides numerous possibilities for information dissemination, collaborative data exchange, and peer-to-peer communication.

$\bullet$ In this study, we quantify network value through network traffic load, as both are generated by the demands of network nodes.
In the everyday understanding, a prevailing consensus suggests \textit{traffic is value}. 
We adopt an assumption of a linear relationship between network value and traffic load, a crucial basis for drawing conclusions. However, in the cube law, it is observed that the scenario is more inclined towards large-scale networks. In such cases, the value generated by network traffic may be significantly greater, potentially breaking through the linear relationship between network traffic and value. We recognize the necessity to embark on future research endeavors aimed at establishing a rigorous and well-defined measurement standard for network value.

\section{Data Availability}
The Milan dataset that support the findings of this study are available from the following website: https://dataverse.harvard.edu/dataset.xhtml?persistentId=doi:10.7910/DVN/EGZHFV and the Gowalla dataset are available from the following website: https://snap.stanford.edu/data/loc-Gowalla.html.

\bibliographystyle{plain}
\bibliography{MetcalfesLaw}
\section*{Author Contributions}
C.W. performed research, designed the methodology and wrote the manuscript with the supervision of C.J.. Y.W. analysed the data and provided the visualization and data presentation.

\section*{Competing Interest}
The authors declare no competing interests.

\section*{Additional Information}
\textbf{Correspondence} and requests for materials should be addressed to Cheng Wang.

\newpage

\begin{table}[t] \renewcommand{\arraystretch}{0.7}
		\captionsetup{type=table, name=Extended Data Table}
	\caption{Some exemplary possible cases of the representative scaling laws under the presumption that there is a linear relationship between the value and traffic load of networks.}
	\label{tab-law-condition} \centering
	
	\resizebox{0.7\linewidth}{!}{	   
		\begin{tabular}{|l|l|}
			\toprule Representative Scaling Laws& Some Exemplary Possible Cases
			\\
			\midrule
			\multirow{6}{*}{Metcalfe's Law: $n^2$} 
			& $\lambda=\Theta(1),\mathcal{s}\geq 0,0\leq\mathcal{d}<1,0\leq \mathcal{i}<1;$
			\\
			& $\lambda=\Theta(n^{1/2}),0\leq\mathcal{s}<1,\mathcal{d}>3/2,\mathcal{i}\geq 0;$
			\\
			& $\lambda=\Theta(n^{1/2}),0\leq \mathcal{s}<1,\mathcal{d}=3/2,\mathcal{i}>1;$
			\\
			& $\lambda=\Theta(n^{1/2}),0\leq \mathcal{s}<1,1<\mathcal{d}<3/2,\mathcal{i}>5/2-\mathcal{d};$
			\\
			& $\lambda=\Theta(n^{1/2}),0\leq \mathcal{s}<1,\mathcal{d}=1,\mathcal{i}\geq 2;$
			\\
			& $\lambda=\Theta(n^{1/2}),0\leq \mathcal{s}<1,0\leq\mathcal{d}<1,\mathcal{i}>3/2.$
			\\
			
			\midrule
			
			\multirow{5}{*}{Sarnoff's Law: $n$} 
			& $\lambda=\Theta(1),\mathcal{s}>2,\mathcal{d}>2,\mathcal{i}\geq0;$
			\\
			& $\lambda=\Theta(1),\mathcal{s}>2,\mathcal{d}=2,\mathcal{i}>1;$
			\\
			& $\lambda=\Theta(1),\mathcal{s}>2,1<\mathcal{d}<2,\mathcal{i}>3-\mathcal{d};$
			\\
			& $\lambda=\Theta(1),\mathcal{s}>2,\mathcal{d}=1,\mathcal{i}\geq 2;$
			\\
			& $\lambda=\Theta(1),\mathcal{s}>2,0\leq\mathcal{d}<1,\mathcal{i}>2.$
			\\
			\midrule
			
			\multirow{5}{*}{Odlyzko's Law: $n\log n$} 
			& $\lambda=\Theta(1),\mathcal{s}=2,\mathcal{d}\geq2,\mathcal{i}\geq0;$
			\\
			& $\lambda=\Theta(1),\mathcal{s}>2,\mathcal{d}=2,0\leq\mathcal{i}<1;$
			\\
			& $\lambda=\Theta(1),\mathcal{s}=2,1\leq\mathcal{d}<2,\mathcal{i}\geq3-\mathcal{d};$
			\\
			& $\lambda=\Theta(1),\mathcal{s}>2,0\leq\mathcal{d}<1,\mathcal{i}=2;$
			\\
			& $\lambda=\Theta(1),\mathcal{s}=2,0\leq\mathcal{d}<1,\mathcal{i}\geq2.$
			\\
			\midrule
			
			\multirow{1}{*}{Cube Law: $n^3$} 
			& $\lambda=\Theta(n),\mathcal{s}\geq 0,0\leq\mathcal{d}<1,0\leq \mathcal{i}<1.$
			\\
			\bottomrule
		\end{tabular}
		
	}
	\vspace{-0.6cm}
\end{table}

\begin{table*}[!t]
	\vspace{-0.3in}
	\renewcommand{\arraystretch}{1.45}
	\centering

	\resizebox{0.9\linewidth}{!}
	{
		\begin{tabular}{ !{\vrule width1.2pt} l!{\vrule width1.2pt}c|c|c|c|c !{\vrule width1.2pt} }

			\toprule
			$\mathcal{d}$ $\backslash$ $\mathcal{s}$ & $\mathcal{s}>2$ & $\mathcal{s}=2$ & $ 1<\mathcal{s}<2 $ & $\mathcal{s}=1$ & $0\leq\mathcal{s}<1$ \\
			\toprule

			$\mathcal{d}>2$ &
%
%
			$\uline{\lambda \cdot \Omega(n)} , \mathcal{i}\geq0. $ &
%
			
			$ \dashuline{\lambda \cdot \Omega(n \log n)} , \mathcal{i}\geq0. $ &
			$ \lambda \cdot \Omega(n^{2-\frac{\mathcal{s}}{2}}) , \mathcal{i}\geq0. $ &
			$ \lambda \cdot \Omega(n^{\frac{3}{2}} / \sqrt{\log n}), \mathcal{i}\geq0. $ &
			
			$ \lambda \cdot \Omega(n^{\frac{3}{2}}), \mathcal{i}\geq0. $ \\
			\midrule
			
			$\mathcal{d}=2$ & $\left\{
			\begin{array}{ll}
%
%
				\uline{\lambda \cdot \Omega(n)}  , \\
				\qquad   \mathcal{i} >1;\\
%
%
				\dashuline{\lambda \cdot \Omega(n \log n)} ,  \\
				\qquad 0\leq\mathcal{i}\leq1.\\
			\end{array}
			\right. $ &
			
%
%
			
			$ \dashuline{\lambda \cdot \Omega(n \log n)}$  ,  $ \mathcal{i}\geq0. $  &
			$ \lambda \cdot \Omega(n^{2-\frac{\mathcal{s}}{2}})  ,  \mathcal{i}\geq0. $  &
			$ \lambda \cdot \Omega(n^{\frac{3}{2}} / \sqrt{\log n}) $ , $ \mathcal{i}\geq0. $ &
			
			$ \lambda \cdot \Omega(n^{\frac{3}{2}}) ,  \mathcal{i}\geq0. $  \\
			\midrule
			
			$\frac{3}{2}<\mathcal{d}<2$ & $\left\{
			\begin{array}{ll}
%
%
				\uline{\lambda \cdot \Omega(n)} , \\
				\qquad  \mathcal{i}>3-\mathcal{d};  \\
%
%
				\dashuline{\lambda \cdot \Omega(n \log n)} , \\
				\qquad  \mathcal{i}=3-\mathcal{d}; \\
				\lambda \cdot \Omega(n^{4-\mathcal{i}-\mathcal{d}})  , \\
				\qquad  1<\mathcal{i}<3-\mathcal{d}; \\
				\lambda \cdot \Omega( n^{3-\mathcal{d}} / \log n) , \\
				\qquad  \mathcal{i}=1; \\
				\lambda \cdot \Omega(n^{3-\mathcal{d}})  , \\
				\qquad  0\leq\mathcal{i}<1.
			\end{array}
			\right. $ &
			$\left\{
			\begin{array}{ll}
%
%
				\dashuline{\lambda \cdot \Omega(n  \log n)} , \\
				\qquad  \mathcal{i}\geq3-\mathcal{d}; \\
				\lambda \cdot \Omega(n^{4-\mathcal{i}-\mathcal{d}})  , \\
				\qquad  1<\mathcal{i}<3-\mathcal{d}; \\
				\lambda \cdot \Omega(n^{3-\mathcal{d}} / \log n) , \\
				\qquad  \mathcal{i}=1; \\
				\lambda \cdot \Omega(n^{3-\mathcal{d}})  , \\
				\qquad  0\leq\mathcal{i}<1.
			\end{array}
			\right. $ &
			$\left\{
			\begin{array}{ll}
				\lambda \cdot \Omega(n^{2-\frac{\mathcal{s}}{2}}) , \\
				\qquad  \mathcal{i}\geq3-\mathcal{d}; \\
				\lambda \cdot \Omega(n^{4-\mathcal{i}-\mathcal{d}})  , \\
				\qquad  1<\mathcal{i}<3-\mathcal{d}; \\
				\lambda \cdot \Omega(n^{3-\mathcal{d}} / \log n) , \\
				\qquad  \mathcal{i}=1; \\
				\lambda \cdot \Omega(n^{3-\mathcal{d}})  , \\
				\qquad  0\leq\mathcal{i}<1.
			\end{array}
			\right. $ &

				$\lambda \cdot \Omega(n^{\frac{3}{2}} / \sqrt{\log n}) ,   \mathcal{i}\geq0. $ &
				$\lambda \cdot \Omega(n^{\frac{3}{2}}),  \mathcal{i}\geq0. $ \\
			\midrule
			
			$\mathcal{d}=\frac{3}{2}$ & $\left\{
			\begin{array}{ll}
%
%
				
				\uline{\lambda \cdot \Omega(n)}  ,   \\
				\qquad \mathcal{i} >\frac{3}{2};\\
%
%
				\dashuline{\lambda \cdot \Omega(n \log n)} , \\
				\qquad \mathcal{i}=\frac{3}{2};  \\
				\lambda \cdot \Omega(n^{\frac{5}{2}-\mathcal{i}}) , \\
				\qquad 1<\mathcal{i}<\frac{3}{2};  \\
				\lambda \cdot \Omega(n^{\frac{3}{2}} / \log n) , \\
				\qquad \mathcal{i}=1;  \\
				\lambda \cdot \Omega(n^{\frac{3}{2}})  ,  \\
				\qquad 0\leq\mathcal{i}<1.\\
			\end{array}
			\right. $ &
			$\left\{
			\begin{array}{ll}
%
%
				\dashuline{\lambda \cdot \Omega(n \log n)}  ,   \\
				\qquad \mathcal{i} \geq\frac{3}{2};\\
				\lambda \cdot \Omega(n^{\frac{5}{2}-\mathcal{i}}) , \\
				\qquad 1<\mathcal{i}<\frac{3}{2};  \\
				\lambda \cdot \Omega(n^{\frac{3}{2}} / \log n) , \\
				\qquad \mathcal{i}=1;  \\
				\lambda \cdot \Omega(n^{\frac{3}{2}})  ,  \\
				\qquad 0\leq\mathcal{i}<1.\\
			\end{array}
			\right. $ &
			$\left\{
			\begin{array}{ll}
				\lambda \cdot \Omega(n^{2-\frac{\mathcal{s}}{2}})  ,   \\
				\qquad  \mathcal{i} \geq\frac{3}{2};\\
				\lambda \cdot \Omega(n^{\frac{5}{2}-\mathcal{i}}) , \\
				\qquad  1<\mathcal{i}<\frac{3}{2};  \\
				\lambda \cdot \Omega(n^{\frac{3}{2}} / \log n) , \\
				\qquad \mathcal{i}=1;  \\
				\lambda \cdot \Omega(n^{\frac{3}{2}})  ,  \\
				\qquad  0\leq\mathcal{i}<1.\\
			\end{array}
			\right. $ &
			$\left\{
			\begin{array}{ll}
				\lambda \cdot \Omega(n^{\frac{3}{2}} / \sqrt{\log n})  ,   \\
				\qquad  \mathcal{i} >1;\\
				\lambda \cdot \Omega(n^{\frac{3}{2}} \sqrt{\log n})  ,  \\
				\qquad  0\leq\mathcal{i}\leq1.\\
			\end{array}
			\right. $ &
			$\left\{
			\begin{array}{ll}
				\lambda \cdot \Omega(n^{\frac{3}{2}}),   \\
				\qquad  \mathcal{i} >1;\\
				\lambda \cdot \Omega(n^{\frac{3}{2}} \log n)  ,  \\
				\qquad  0\leq\mathcal{i}\leq1.\\
			\end{array}
			\right. $ \\
			\midrule
			
			$1<\mathcal{d}<\frac{3}{2}$ & $\left\{
			\begin{array}{ll}
%
%
				
				\uline{\lambda \cdot \Omega(n)}  ,  \\
				\qquad  \mathcal{i}>3-\mathcal{d}; \\
%
%
				\dashuline{\lambda \cdot \Omega(n \log n)} , \\
				\qquad \mathcal{i}=3-\mathcal{d}; \\
				\lambda \cdot \Omega(n^{4-\mathcal{i}-\mathcal{d}}) , \\
				\qquad  1<\mathcal{i}<3-\mathcal{d};  \\
				\lambda \cdot \Omega(n^{3-\mathcal{d}} / \log n) , \\
				\qquad  \mathcal{i}=1; \\
				\lambda \cdot \Omega(n^{3-\mathcal{d}})  ,  \\
				\qquad 0\leq\mathcal{i}<1.\\
			\end{array}
			\right. $ &
			$\left\{
			\begin{array}{ll}
%
%
				\dashuline{\lambda \cdot \Omega(n \log n)}  ,  \\
				\qquad \mathcal{i}\geq3-\mathcal{d}; \\
				\lambda \cdot \Omega(n^{4-\mathcal{i}-\mathcal{d}}) , \\
				\qquad 1<\mathcal{i}<3-\mathcal{d};  \\
				\lambda \cdot \Omega(n^{3-\mathcal{d}} / \log n) , \\
				\qquad \mathcal{i}=1; \\
				\lambda \cdot \Omega(n^{3-\mathcal{d}})  ,  \\
				\qquad  0\leq\mathcal{i}<1.\\
			\end{array}
			\right. $ &
			$\left\{
			\begin{array}{ll}
				\lambda \cdot \Omega(n^{2-\frac{\mathcal{s}}{2}})  ,  \\
				\qquad  \mathcal{i}\geq3-\mathcal{d}; \\
				\lambda \cdot \Omega(n^{4-\mathcal{i}-\mathcal{d}}) , \\
				\qquad  1<\mathcal{i}<3-\mathcal{d};  \\
				\lambda \cdot \Omega(n^{3-\mathcal{d}} / \log n) , \\
				\qquad  \mathcal{i}=1; \\
				\lambda \cdot \Omega(n^{3-\mathcal{d}})  ,  \\
				\qquad 0\leq\mathcal{i}<1.\\
			\end{array}
			\right. $ &
			$\left\{
			\begin{array}{ll}
				\lambda \cdot \Omega(n^{\frac{3}{2}} / \sqrt{\log n})  ,  \\
				\qquad \mathcal{i}>\frac{5}{2}-\mathcal{d}; \\
				\lambda \cdot \Omega(n^{\frac{3}{2}} \sqrt{\log n}) , \\
				\qquad \mathcal{i}=\frac{5}{2}-\mathcal{d}; \\
				\lambda \cdot \Omega(n^{4-\mathcal{i}-\mathcal{d}}) , \\
				\qquad 1<\mathcal{i}<\frac{5}{2}-\mathcal{d};  \\
				\lambda \cdot \Omega(n^{3-\mathcal{d}} / \log n) , \\
				\qquad \mathcal{i}=1; \\
				\lambda \cdot \Omega(n^{3-\mathcal{d}})  ,  \\
				\qquad 0\leq\mathcal{i}<1.\\
			\end{array}
			\right. $ &
			$\left\{
			\begin{array}{ll}
				
				\lambda \cdot \Omega(n^{\frac{3}{2}}),  \\
				\qquad \mathcal{i}>\frac{5}{2}-\mathcal{d}; \\
				\lambda \cdot \Omega(n^{\frac{3}{2}} \cdot \log n) , \\
				\qquad \mathcal{i}=\frac{5}{2}-\mathcal{d}; \\
				\lambda \cdot \Omega(n^{4-\mathcal{i}-\mathcal{d}}) , \\
				\qquad 1<\mathcal{i}<\frac{5}{2}-\mathcal{d};  \\
				\lambda \cdot \Omega(n^{3-\mathcal{d}} / \log n) , \\
				\qquad \mathcal{i}=1; \\
				\lambda \cdot \Omega(n^{3-\mathcal{d}})  ,  \\
				\qquad  0\leq\mathcal{i}<1.\\
			\end{array}
			\right. $ \\
			\midrule
			
			$\mathcal{d}=1$ & $\left\{
			\begin{array}{ll}
				
%
%
				\uline{\lambda \cdot \Omega(n)}  ,   \\
				\qquad  \mathcal{i}\geq 2;\\
				\lambda \cdot \Omega(n^{3-\mathcal{i}} / \log n) , \\
				\qquad  1<\mathcal{i}<2;  \\
				\lambda \cdot \Omega(n^{2} / (\log n)^2) , \\
				\qquad  \mathcal{i}=1;  \\
				\lambda \cdot \Omega(n^{2} / \log n)  ,  \\
				\qquad 0\leq\mathcal{i}<1.\\
			\end{array}
			\right. $ &
			$\left\{
			\begin{array}{ll}
%
%
				\dashuline{\lambda \cdot \Omega(n \log n)}  ,   \\
				\qquad \mathcal{i}\geq 2;\\
				\lambda \cdot \Omega(n^{3-\mathcal{i}} / \log n) , \\
				\qquad 1<\mathcal{i}<2;  \\
				\lambda \cdot \Omega(n^{2} / (\log n)^2) , \\
				\qquad  \mathcal{i}=1;  \\
				\lambda \cdot \Omega(n^{2} / \log n)  ,  \\
				\qquad 0\leq\mathcal{i}<1.\\
			\end{array}
			\right. $ &
			$\left\{
			\begin{array}{ll}
				\lambda \cdot \Omega(n^{2-\frac{\mathcal{s}}{2}})  ,   \\
				\qquad \mathcal{i}\geq 2;\\
				\lambda \cdot \Omega(n^{3-\mathcal{i}} / \log n) , \\
				\qquad 1<\mathcal{i}<2;  \\
				\lambda \cdot \Omega(n^{2} / (\log n)^2) , \\
				\qquad  \mathcal{i}=1;  \\
				\lambda \cdot \Omega(n^{2} / \log n)  ,  \\
				\qquad 0\leq\mathcal{i}<1.\\
			\end{array}
			\right. $ &
			$\left\{
			\begin{array}{ll}
				\lambda \cdot \Omega(n^{\frac{3}{2}} / \sqrt{\log n})  ,   \\
				\qquad \mathcal{i}\geq 2;\\
				\lambda \cdot \Omega(n^{3-\mathcal{i}} / \log n) , \\
				\qquad  1<\mathcal{i}<2;  \\
				\lambda \cdot \Omega(n^{2} / (\log n)^2) , \\
				\qquad  \mathcal{i}=1;  \\
				\lambda \cdot \Omega(n^{2} / \log n)  ,  \\
				\qquad 0\leq\mathcal{i}<1.\\
				
			\end{array}
			\right. $ &
			$\left\{
			\begin{array}{ll}
				
				\lambda \cdot \Omega(n^{\frac{3}{2}}),   \\
				\qquad \mathcal{i}\geq 2;\\
				\lambda \cdot \Omega(n^{3-\mathcal{i}} / \log n) , \\
				\qquad  1<\mathcal{i}<2;  \\
				\lambda \cdot \Omega(n^{2} / (\log n)^2) , \\
				\qquad  \mathcal{i}=1;  \\
%
%
				
				\lambda \cdot \Omega(n^{2} / \log n)  ,  \\
				\qquad  0\leq\mathcal{i}<1.\\
			\end{array}
			\right. $ \\
			\hline
			
			$0\leq\mathcal{d}<1$ & $\left\{
			\begin{array}{ll}
				
%
%
				\uline{\lambda \cdot \Omega(n)}  ,   \\
%
%
				\dashuline{\lambda \cdot \Omega(n \log n)} , \\
				\qquad  \mathcal{i}=2;  \\
				\lambda \cdot \Omega(n^{3-\mathcal{i}}) , \\
				\qquad  1<\mathcal{i}<2;  \\
				\lambda \cdot \Omega(n^{2} / \log n) , \\
				\qquad  \mathcal{i}=1;  \\
				
				\uuline{\lambda \cdot \Omega(n^{2})}, \\
				\qquad  0\leq\mathcal{i}<1.\\
			\end{array}
			\right. $ &
			$\left\{
			\begin{array}{ll}
				
			\dashuline{\lambda \cdot \Omega(n \log n)},\\
				\qquad  \mathcal{i}\geq2;\\
				\lambda \cdot \Omega(n^{3-\mathcal{i}}) , \\
				\qquad  1<\mathcal{i}<2;  \\
				\lambda \cdot \Omega(n^{2} / \log n) , \\
				\qquad  \mathcal{i}=1;  \\
				\uuline{ \lambda \cdot \Omega(n^{2})},  \\
				\qquad  0\leq\mathcal{i}<1.\\
			\end{array}
			\right. $ &
			$\left\{
			\begin{array}{ll}
				\lambda \cdot \Omega(n^{2-\frac{\mathcal{s}}{2}})  ,   \\
				\qquad \mathcal{i}\geq2;\\
				\lambda \cdot \Omega(n^{3-\mathcal{i}}) , \\
				\qquad  1<\mathcal{i}<2;  \\
				\lambda \cdot \Omega(n^{2} / \log n) , \\
				\qquad  \mathcal{i}=1;  \\
				\uuline{ \lambda \cdot \Omega(n^{2})}, \\
				\qquad  0\leq\mathcal{i}<1.\\
			\end{array}
			\right. $ &
			$\left\{
			\begin{array}{ll}
				\lambda \cdot \Omega(n^{\frac{3}{2}} / \sqrt{\log n})  ,   \\
				\qquad  \mathcal{i}>\frac{3}{2};\\
				\lambda \cdot \Omega(n^{\frac{3}{2}} \sqrt{\log n}) , \\
				\qquad  \mathcal{i}=\frac{3}{2};  \\
				\lambda \cdot \Omega(n^{3-\mathcal{i}}) , \\
				\qquad  1<\mathcal{i}<\frac{3}{2};  \\
				\lambda \cdot \Omega(n^{2} / \log n) , \\
				\qquad  \mathcal{i}=1;  \\
				
			\uuline{\lambda \cdot \Omega(n^{2})},  \\
				\qquad  0\leq\mathcal{i}<1.\\
			\end{array}
			\right. $ &
			$\left\{
			\begin{array}{ll}
				\lambda \cdot \Omega(n^{\frac{3}{2}}),   \\
				\qquad  \mathcal{i}>\frac{3}{2};\\
				\lambda \cdot \Omega(n^{\frac{3}{2}} \cdot \log n) , \\
				\qquad  \mathcal{i}=\frac{3}{2};  \\
				\lambda \cdot \Omega(n^{3-\mathcal{i}}) , \\
				\qquad  1<\mathcal{i}<\frac{3}{2};  \\
				\lambda \cdot \Omega(n^{2} / \log n) , \\
				\qquad  \mathcal{i}=1;  \\
				\uuline{\lambda \cdot \Omega(n^{2})},  \\

				\qquad  0\leq\mathcal{i}<1.\\
			\end{array}
			\right. $ \\
			\bottomrule
			\multicolumn{6}{c}{

				\uuline{\textbf{Metcalfe's Law}}\qquad \qquad \quad 
				\uline{\textbf{Sarnoff's Law}}\qquad\qquad \quad 
				\dashuline{\textbf{Odlyzko's Law}}\qquad 		\qquad \quad 
				When $\lambda=\Theta(1)$, \textbf{cube law} cannot be achieved.
				
			}
		\end{tabular}

	}
	\captionsetup{type=table, name=Extended Data Table}
	\caption{\centering The Results on the Traffic Load $\mathrm{L}_{\mathrm{N}}^{\ast}$ and Labeling of Scaling Laws under $\lambda=\Theta(1)$. Different theoretical results derived from the distribution above.}	\label{G-tab-total}
\end{table*}

\begin{figure*}[t]
	\centering
	
	\vspace{3pt}
	\begin{minipage}{0.32\linewidth}
		\subfloat[$0\leq \mathcal{d}<1,\mathcal{s}>2$]{\includegraphics[width=\textwidth]{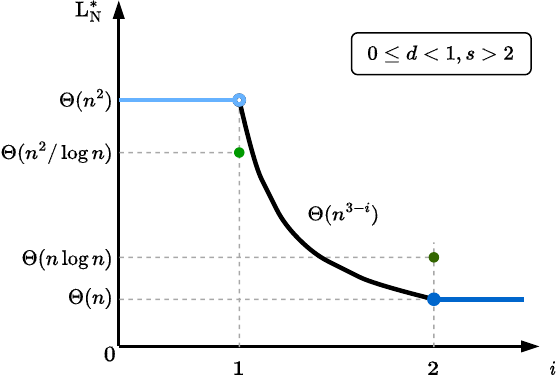}}
		
	\end{minipage}
	\begin{minipage}{0.32\linewidth}
		\subfloat[$0\leq \mathcal{d}<1,\mathcal{s}=2$]{\includegraphics[width=\textwidth]{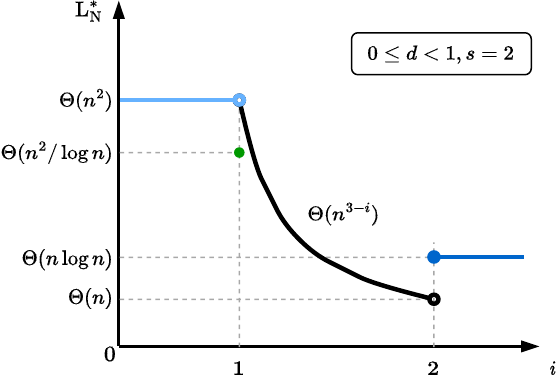}}
	\end{minipage}
	\begin{minipage}{0.32\linewidth}
		\subfloat[$0\leq \mathcal{d}<1,1<\mathcal{s}<2$]{\includegraphics[width=\textwidth]{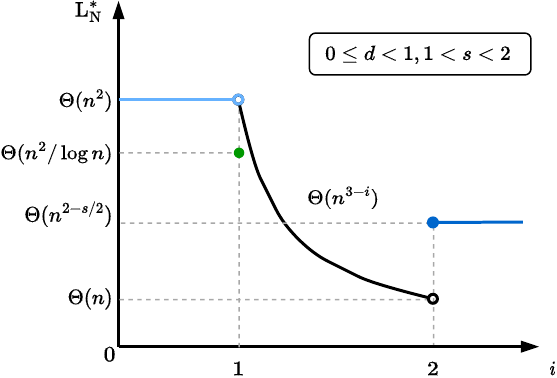}}
	\end{minipage}
	\begin{minipage}{0.35\linewidth}
		\vspace{10pt}
		\subfloat[$0\leq \mathcal{d}<1,\mathcal{s}=2$]{\includegraphics[width=\textwidth]{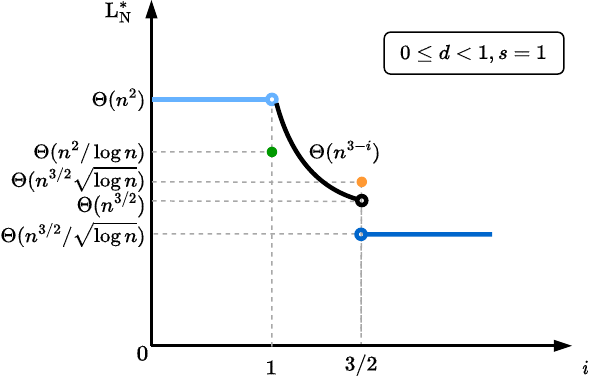}}
	\end{minipage}
	\hspace{10mm}
	\begin{minipage}{0.35\linewidth}
		\vspace{10pt}
		\subfloat[$0\leq \mathcal{d}<1,0<\mathcal{s}<1$]{\includegraphics[width=\textwidth]{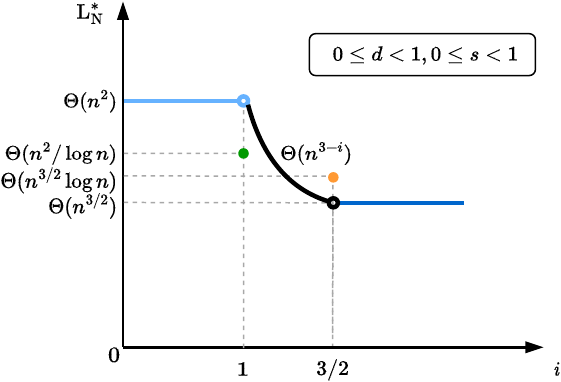}}
	\end{minipage}
	\captionsetup{type=figure, name=Extended Data Figure}
	\caption{The order of low bounds on the derived traffic load, under $\lambda=\Theta(1)$ and $0\le \mathcal{d}<1$.}\label{pic-understand}
	\label{thr}
\end{figure*}

\begin{table}[t]
	\centering

	\renewcommand\arraystretch{1.2}
	\resizebox{0.4\linewidth}{!}{	\begin{tabular}{|l ||l|}
			\toprule
			$\mathcal{d}$ & $ G\left(\mathcal{i}, \mathcal{d}\right) $\\
			\midrule
			
			$\mathcal{d}>\frac{3}{2}$ &
			$ \Theta(1) $  , $ \mathcal{i}\geq0.$ \\
			\midrule
			$\mathcal{d}=\frac{3}{2}$ &
			$\left\{
			\begin{array}{ll}
				\Theta(1)  , &  \mathcal{i} >1; \\
				\Theta(\log n)  , &  0\leq\mathcal{i}\leq1.
			\end{array}
			\right. $ \\
			\midrule
			$1<\mathcal{d}<\frac{3}{2}$ &
			$\left\{
			\begin{array}{ll}
				\Theta(1) , &  \mathcal{i}>\frac{5}{2}-\mathcal{d};  \\
				\Theta(\log n) , &  \mathcal{i}=\frac{5}{2}-\mathcal{d}; \\
				\Theta(n^{\frac{5}{2}-\mathcal{i}-\mathcal{d}})  , &  1<\mathcal{i}<\frac{5}{2}-\mathcal{d}; \\
				\Theta(n^{\frac{3}{2}-\mathcal{d}} / \log n) , &  \mathcal{i}=1; \\
				\Theta(n^{\frac{3}{2}-\mathcal{d}})  , &  0\leq\mathcal{i}<1.
			\end{array}
			\right. $ \\
			\midrule
			$\mathcal{d}=1$ &
			$\left\{
			\begin{array}{ll}
				\Theta(1)  , &  \mathcal{i}\geq\frac{3}{2};\\
				\Theta(n^{\frac{3}{2}-\mathcal{i}}/ \log n) , &  1<\mathcal{i}<\frac{3}{2};  \\
				\Theta(n^{\frac{1}{2}} / (\log n)^2) , &  \mathcal{i}=1;  \\
				\Theta(n^{\frac{1}{2}}/ \log n)  , &  0\leq\mathcal{i}<1.
			\end{array}
			\right. $ \\
			\midrule
			$0\leq\mathcal{d}<1$ &
			$\left\{
			\begin{array}{ll}
				\Theta(1)  , &  \mathcal{i}>\frac{3}{2}; \\
				\Theta(\log n)  , &  \mathcal{i}=\frac{3}{2}; \\
				\Theta(n^{\frac{3}{2}-\mathcal{i}})  , &  1<\mathcal{i}<\frac{3}{2}; \\
				\Theta(n^{\frac{1}{2}}/\log n) , &  \mathcal{i}=1;  \\
				\Theta(n^{\frac{1}{2}})  , &  0\leq\mathcal{i}<1.
			\end{array}
			\right. $ \\
			\bottomrule
		\end{tabular}
	}
	\captionsetup{type=table, name=Extended Data Table}
	\caption{The Value of $G(\mathcal{i},\mathcal{d})$}	\label{tab-G}
\end{table}

\begin{table}[t]
	\centering

	\renewcommand\arraystretch{1.2}
	\resizebox{0.4\linewidth}{!}{	\begin{tabular}{|l ||l|}
			\toprule
			$\mathcal{d}$ & $ W\left(\mathcal{i}, \mathcal{d}\right) $\\
			\midrule
			$\mathcal{d}>2$ &
			$ \Theta(n) $  , $ \mathcal{i}\geq0.$ \\
			\hline
			$\mathcal{d}=2$ &
			$\left\{
			\begin{array}{ll}
				\Theta(n)  , &  \mathcal{i} >1; \\
				\Theta(n \cdot \log n)  , &  0\leq\mathcal{i}\leq1.
			\end{array}
			\right. $ \\
			\midrule
			$1<\mathcal{d}<2$ &
			$\left\{
			\begin{array}{ll}
				\Theta(n) , &  \mathcal{i}>3-\mathcal{d};  \\
				\Theta(n \cdot \log n) , &  \mathcal{i}=3-\mathcal{d}; \\
				\Theta(n^{4-\mathcal{i}-\mathcal{d}})  , &  1<\mathcal{i}<3-\mathcal{d}; \\
				\Theta(n^{3-\mathcal{d}} / \log n) , &  \mathcal{i}=1; \\
				\Theta(n^{3-\mathcal{d}})  , &  0\leq\mathcal{i}<1.
			\end{array}
			\right. $ \\
			\midrule
			$\mathcal{d}=1$ &
			$\left\{
			\begin{array}{ll}
				\Theta(n)  , &  \mathcal{i}\geq2;\\
				\Theta(n^{3-\mathcal{i}}/ \log n) , &  1<\mathcal{i}<2;  \\
				\Theta(n^{2} / (\log n)^2) , &  \mathcal{i}=1;  \\
				\Theta(n^{2}/ \log n)  , &  0\leq\mathcal{i}<1.
			\end{array}
			\right. $ \\
			\midrule
			$0\leq\mathcal{d}<1$ &
			$\left\{
			\begin{array}{ll}
				\Theta(n)  , &  \mathcal{i}>2; \\
				\Theta(n \cdot \log n)  , &  \mathcal{i}=2; \\
				\Theta(n^{3-\mathcal{i}})  , &  1<\mathcal{i}<2; \\
				\Theta(n^{2}/\log n) , &  \mathcal{i}=1;  \\
				\Theta(n^{2})  , &  0\leq\mathcal{i}<1.
			\end{array}
			\right. $ \\
			\bottomrule
	\end{tabular}}
	\captionsetup{type=table, name=Extended Data Table}
	\caption{The Number of All Destinations, $W\left(\mathcal{i}, \mathcal{d}\right)$}\label{tab-W-gamma-varphi}
\end{table}

\begin{figure}[t]
	\centering
	\includegraphics[width=0.48\textwidth]{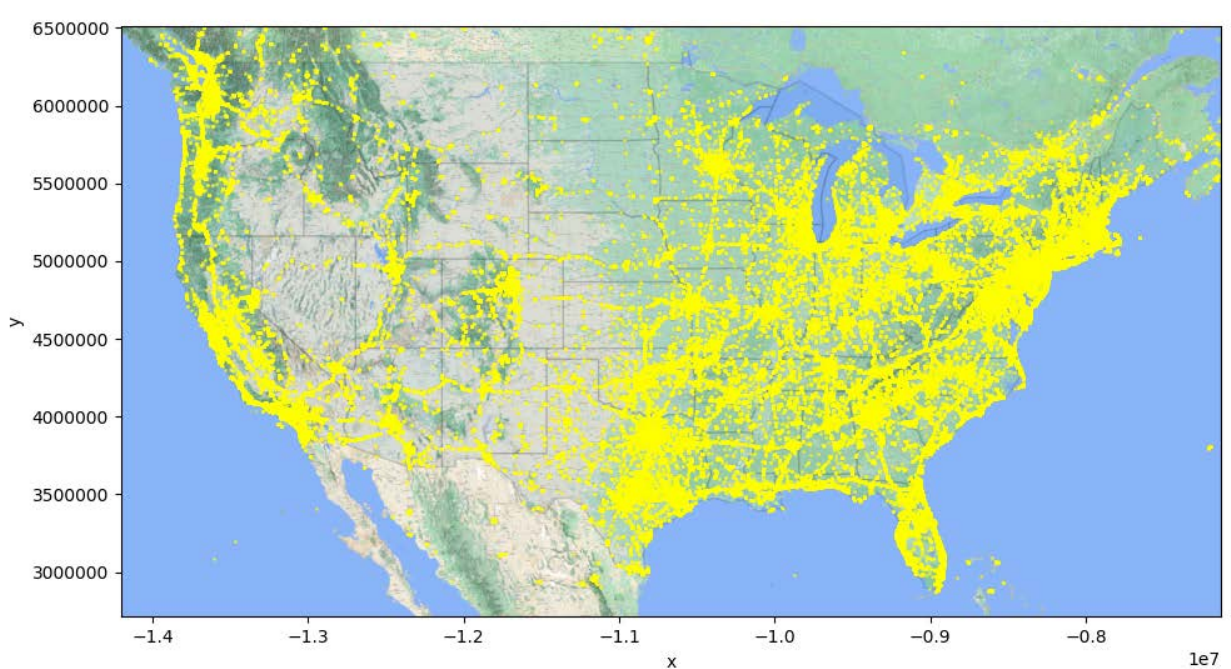}
	\captionsetup{type=figure, name=Extended Data Figure}
	\caption{The user distribution map of Gowalla illustrates the geographical spread plotted against latitude and longitude. Utilizing this information, we pinpointed two specific areas for closer examination in Figure \ref{fig_gowalla} to assess the rationale behind setting $\mathcal{g}=0$ in this study.}
	\label{fig:final_map}
\end{figure}

\begin{figure}[t]
	\centering
	\begin{minipage}{0.4\linewidth}
		\subfloat[Bloomington in Illinois]{\includegraphics[width=0.8\textwidth]{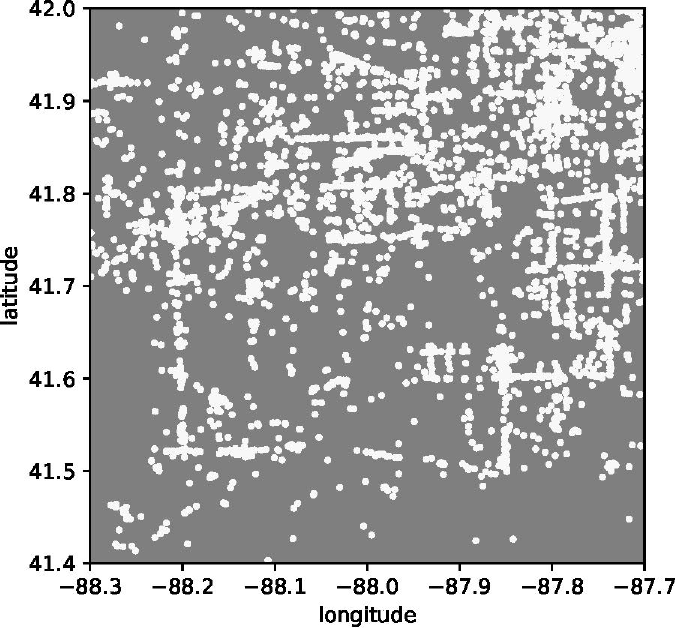}}
		
	\end{minipage}
	\begin{minipage}{0.5\linewidth}
		\subfloat[San Francisco in California]{\includegraphics[width=0.8\textwidth]{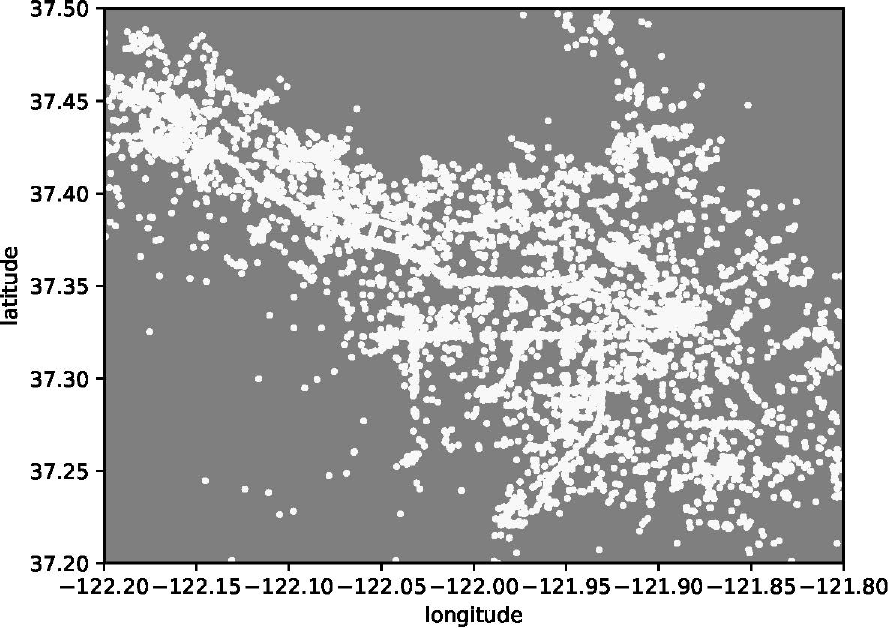}}
	\end{minipage}
	\vspace{-0.1in}
	\captionsetup{type=figure, name=Extended Data Figure}
	\caption{The geographical position of Gowalla Users in Illinois and California, respectively.}
	\label{fig_gowalla}
	\vspace{-0.5cm}
\end{figure}

\begin{figure*}[t]
	\centering
	\subfloat[Sarnoff's Function $y=an+b$]{\begin{minipage}{0.48\linewidth}
			\includegraphics[width=\textwidth]{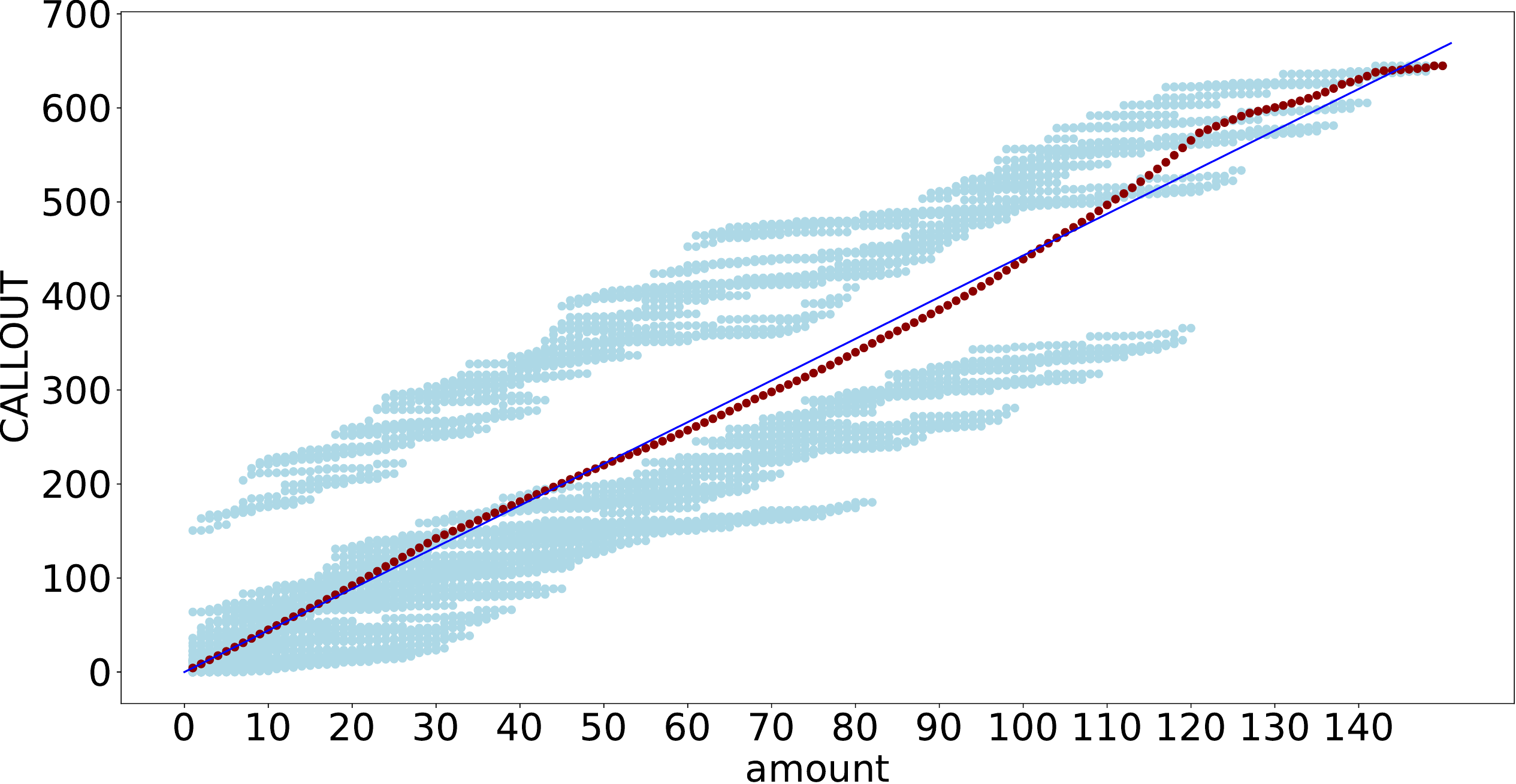}
	\end{minipage}}
\subfloat[Odlyzko's Function $y=an\ln n+bn +c$]{\begin{minipage}{0.48\linewidth}
		\includegraphics[width=\textwidth]{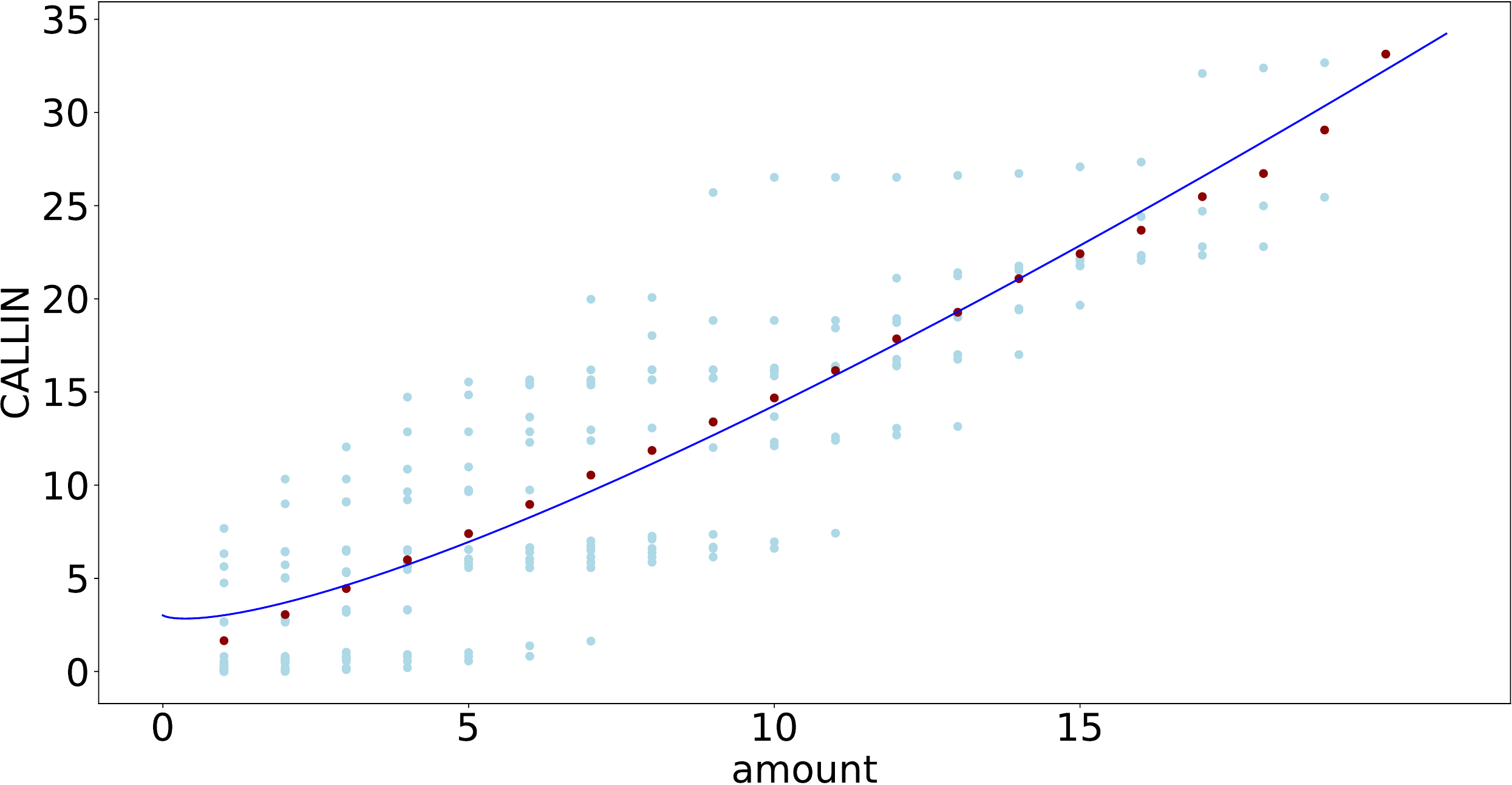}
\end{minipage}}

\subfloat[Metcalfe's Function $y=an^2+bn +c$]{\begin{minipage}{0.48\linewidth}
	\includegraphics[width=\textwidth]{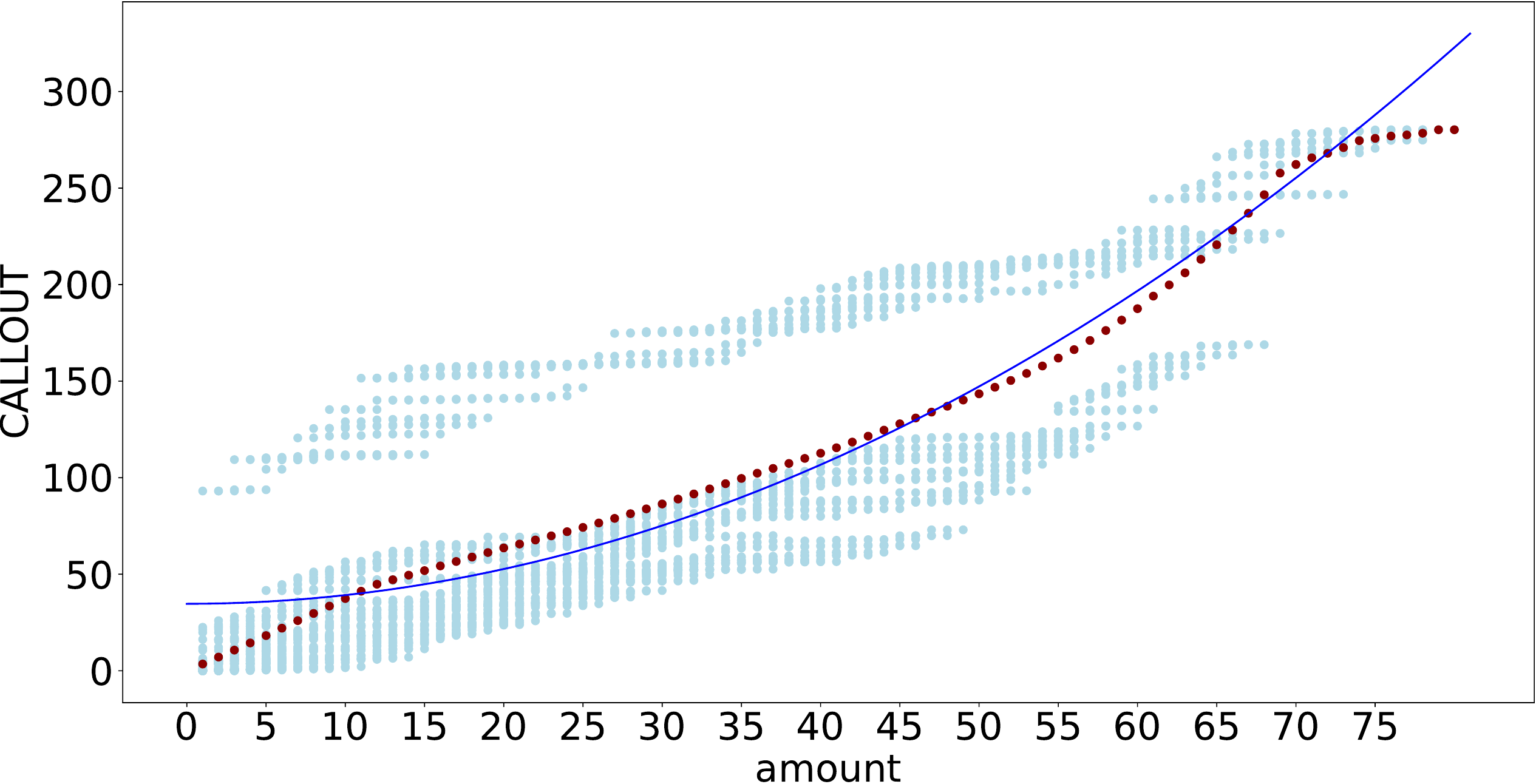}
\end{minipage}}
\subfloat[Cube Function $y=an^3+bn^2 +cn+d$]{\begin{minipage}{0.48\linewidth}
\includegraphics[width=\textwidth]{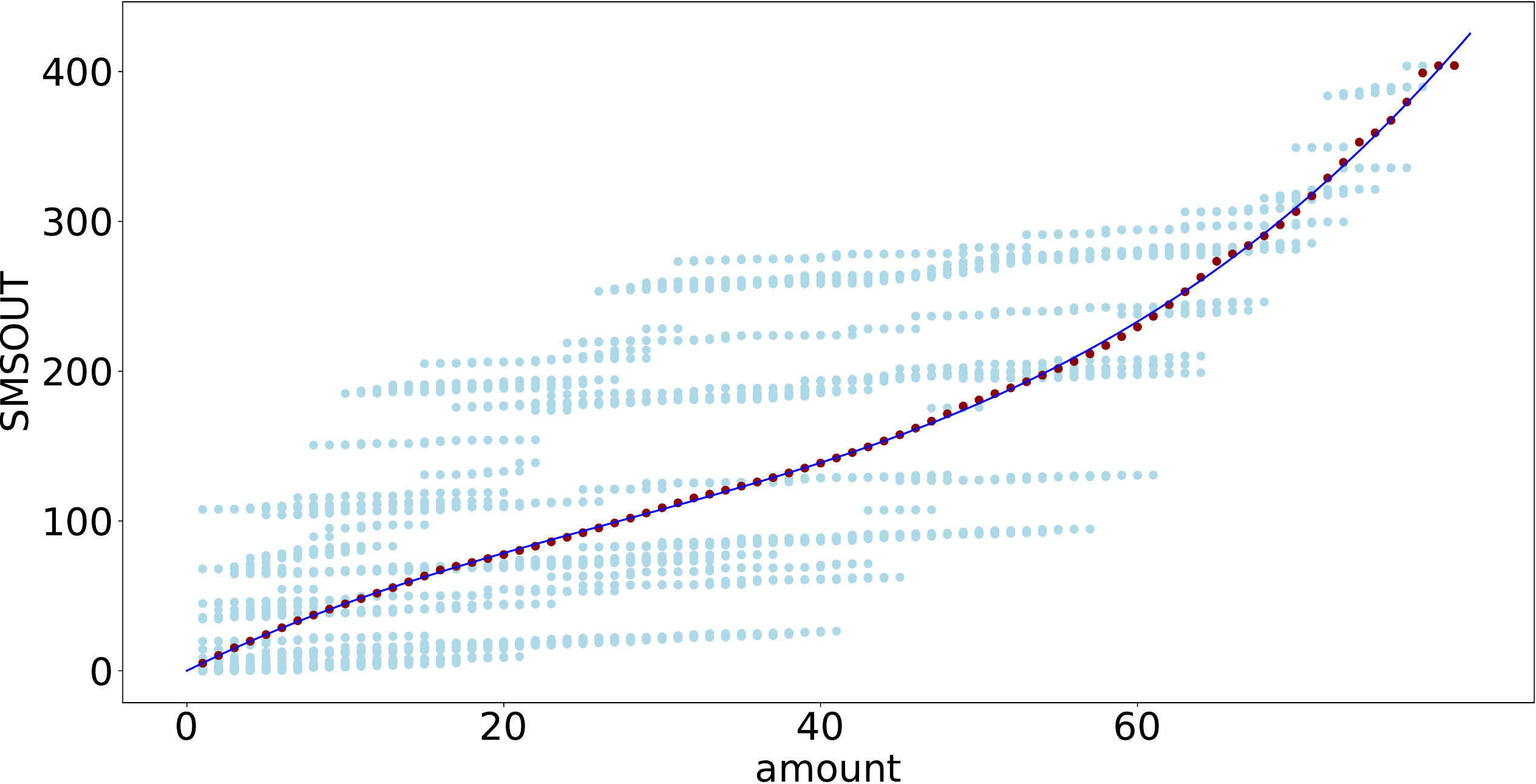}
\end{minipage}}
\captionsetup{type=table, name=Extended Data Figure}
\caption{The other empirical curve fitting results of Metcalfe's law and its variants, which can be seen as $y=an+b,y=an\ln n +bn+c, y=an^2+bn+c, y=an^3+bn^2+cn+d$, corresponding to Sarnoff's function, Odlyzko's function, Metcalfe's function, and Cube function, respectively.}
\label{pic-fit2}
\vspace{-0.1in}
\end{figure*}

\begin{table*}[t] \renewcommand{\arraystretch}{1.2}
	
	 \centering
	
	\resizebox{0.9\linewidth}{!}{    
		\begin{tabular}{|l|l|l|l|l|}
			\toprule & Given Day&Number of Nodes $n$ & Interaction Types & Value Function              \\ \midrule
			\multirow{2}{*}{Sarnoff's function}  &3& 110            &Call-out& $\mathrm{L}^*_{\mathrm{N}}(n)=2.34n$                     \\ \cmidrule{2-5}
			&4&150             &Call-out                        & $\mathrm{L}^*_{\mathrm{N}}(n)=4.43n$                     \\ \midrule
			\multirow{2}{*}{Odlyzko's function}  & 4&20         &Call-out     & $\mathrm{L}^*_{\mathrm{N}}(n)=-0.5n\ln n+4.71n-2.99$       \\ \cmidrule{2-5}
			&     6 & 20             & Call-in                  & $\mathrm{L}^*_{\mathrm{N}}(n)=0.48n\ln n+3.02$       \\ \midrule
			\multirow{2}{*}{Metcalfe's function} & 25&80           &Call-out & $\mathrm{L}^*_{\mathrm{N}}(n)=0.094n^2+74.65$            \\ \cmidrule{2-5}
			&  25   & 160                & SMS-out             & $\mathrm{L}^*_{\mathrm{N}}(n)=0.034n^2+101.14$            \\ \midrule
			\multirow{2}{*}{Cube function}       &12& 60                &Call-in& $\mathrm{L}^*_{\mathrm{N}}(n)=0.0021n^3-0.129n^2+5.82n$  \\ \cmidrule{2-5}
			&     8  &       80                      &SMS-out& $\mathrm{L}^*_{\mathrm{N}}(n)=0.0011n^3-0.088n^2+5.25n$ \\ \bottomrule
		\end{tabular}
		
	}
	\captionsetup{type=table, name=Extended Data Table}
	\caption{Empirical Result of Metcalfe's Law and its Variants}
	\label{tab-law-fit}
	\vspace{-0.3cm}
\end{table*}

\end{document}